\documentclass[letterpaper, 10 pt, conference]{ieeeconf}  

\IEEEoverridecommandlockouts                              

\makeatletter
\newcommand{\removelatexerror}{\let\@latex@error\@gobble}
\makeatother
\usepackage[T1]{fontenc}
\usepackage[utf8]{inputenc}
\usepackage{cite}
\usepackage{xcolor}
\usepackage{dsfont}
\usepackage[super]{nth}

\usepackage{amsmath,amssymb,amsfonts,amsthm}
\usepackage{graphicx}
\usepackage{epstopdf}
\usepackage[font=scriptsize]{caption}
\usepackage{subcaption}
\usepackage{epsfig}
\usepackage{cite}
\let\labelindent\relax
\usepackage{enumitem}

\usepackage{amsmath,amssymb,amsfonts}
\usepackage{color,amsthm} 

\newtheorem{remark}{Remark}
\newtheorem{theorem}{Theorem}

\newtheorem{corollary}{Corollary}
\newtheorem{proposition}{Proposition}
\newtheorem{definition}{Definition}
\newtheorem{problem}{Problem}

\newtheorem{assumption}{Assumption}

\allowdisplaybreaks

\title{\LARGE \bf
Encirclement Guaranteed Finite-Time Capture \\ against Unknown Evader Strategies
}

\author{Dinesh Patra$^{1}$, Prajakta Surve$^{2}$, Ashish R. Hota$^{1}$, Shaunak D. Bopardikar$^{2}$%
\thanks{$^{1}$D. Patra and A. R. Hota are with the Department of Electrical Engineering, IIT Kharagpur, India 
(e-mail: dinesh.patra912@kgpian.iitkgp.ac.in; ahota@ee.iitkgp.ac.in). Dinesh Patra was supported by the IEEE CSS Graduate Collaboration Fellowship, 2024 for this work.}%
\thanks{$^{2}$P. Surve and S. D. Bopardikar are with the Department of Electrical and Computer Engineering, Michigan State University, USA 
(e-mail: survepra@msu.edu; shaunak@egr.msu.edu). Their work is supported in part by the Aerospace Systems Technology Research and Assessment (ASTRA) Aerospace Technology Development and Testing (ATDT) program at AFRL under contract number FA8650-21-D-2602 as well as the AFOSR LRIR \# 24RQCOR002. DISTRIBUTION STATEMENT  A. Distribution is unlimited. AFRL-2026-1226, Cleared 03-13-2026.}}

\begin{document}
\maketitle
\thispagestyle{empty}
\pagestyle{empty}
\begin{abstract}
We consider a pursuit-evasion scenario involving a group of pursuers and a single evader in a two-dimensional unbounded environment. The pursuers aim to capture the evader in finite time while ensuring that the evader remains enclosed within the convex hull of their positions until capture, without knowledge of the evader's heading angle. Prior works have addressed the problem of encirclement and capture separately in different contexts. In this letter, we present a class of strategies for the pursuers that guarantee capture in finite time while maintaining encirclement, irrespective of the evader's strategy. Furthermore, we derive an upper bound on the time to capture. Numerical results highlight the effectiveness of the proposed framework against a range of evader strategies.
\end{abstract}

\section{Introduction} \label{section:introduction}

Consider a situation where an unknown intruder enters a secure confinement zone formed by a team of ground robots deployed for surveillance. Once the intruder is inside the zone, the robots aim to keep it confined and approach it progressively to assess its intent; moreover, timely neutralization may be required if a threat is detected. This motivates a pursuit-evasion problem in which a team of pursuers must (i) \,\emph{maintain encirclement}, that is, confine the evader inside the convex hull of their positions at all times, and (ii) \,\emph{guarantee capture in finite time}, irrespective of and without knowledge of the strategy deployed by the evader. 

Encirclement-guaranteed pursuit problems have emerged as an important research direction in recent years. Early studies on static and dynamic target encirclement \cite{ref_early_encircle_lion_jankovic1978man,ref_early_kopparty2005framework,ref_early_encircle_marasco2012model,ref_early_bopardikar2014k} focused on surveillance and defense scenarios, where the objective is to restrict the opponent's motion and approach it progressively, thereby minimizing escalation rather than attempting immediate capture. When the evader is faster \cite{ref_encirclement_faster_chen2016multi,ref_encirclement_faster_garcia2021cooperative} or has equal speed \cite{ref_early_bopardikar2014k}, encirclement often becomes a necessary ingredient for capture. However, when the evader is slower, encirclement and capture are typically treated as distinct objectives \cite{ref_slow_encirclement_wang2023distributed,ref_RMPC_patra20253d}.

A number of recent approaches leverage geometric partitions (e.g., Voronoi-based and area minimization ideas) to reason about encirclement and/or capture in different contexts \cite{ref_Vornoi_area_huang2011guaranteed,ref_Vornoi_area_zhou2016cooperative,ref_slow_encirclement_wang2023distributed,ref_ratnoo_kumar2025cooperative}. For instance, \cite{ref_ratnoo_kumar2025cooperative} derives a worst-case upper bound on the capture time under a specific evader policy, where the evader moves toward the Chebyshev center of its Voronoi cell but relies on a centralized server to collect the position and velocity inputs of all the players.

A related stream of work \cite{Lama2024, Pierson2017} has examined the herding problem, where herders steer non-adversarial agents to a specified goal region using predictable repulsion-based responses. In contrast, our goal is to determine strategies for the pursuers to capture the evader rather than steer it, and no physical repulsion is exerted on the evader in this work.

In a complementary direction, partitioning-based robust model predictive control (MPC) formulations treat the evader input as an uncertain disturbance and aim to enforce encirclement throughout the pursuit \cite{ref_RMPC_patra20253d,ref_slow_RMPC_encirclement_wang2025distributed}. However, these methods typically require online optimization and do not provide rigorous capture guarantees. Game-theoretic approaches \cite{cd5_deng2020multi,ref_encirclement_faster_garcia2021cooperative,cd1,cd2,ref_pd_3d_garcia2022optimal} offer principled strategies when the players' objectives are known. However, they are not directly applicable when the evader's objective (and hence its policy) is not known to the pursuers.

Thus, while prior works provide guarantees on encirclement or capture under particular assumptions in isolation, we are not aware of any existing work that provides simultaneous guarantees of (i) encirclement against unknown evader strategies and (ii) finite-time capture with an explicit analytic upper bound. In this work, we address this research gap with a simple decentralized strategy and a Lyapunov-based analysis. Our main contributions are summarized below.
\begin{enumerate}
    \item We first associate the idea of encirclement of the evader in the convex hull of pursuers with triangulations, where the areas of the triangles formed by the evader and the pursuers along the edges of the convex hull serve as a key quantity of interest. 
    \item We propose a class of strategies for the pursuers that are necessary and sufficient to guarantee encirclement. We then derive conditions on the agents' speeds and decentralized control inputs that are sufficient to guarantee encirclement of the evader, irrespective of its input.
    \item We show that a subset of the encirclement-guaranteeing strategies is, in fact, sufficient to ensure finite-time capture when the evader is strictly slower than the pursuers, and we derive an explicit worst-case upper bound on the capture time using Lyapunov analysis.\footnote{In contrast, the capture guarantees in \cite{ref_slow_RMPC_encirclement_wang2025distributed} are empirical in nature.}
\end{enumerate}
Finally, we validate the proposed strategy against several representative evader policies via numerical simulations. 

\section{Problem Statement}\label{section:problem_description}

We consider a pursuit-evasion scenario in the unbounded planar domain $\mathbb{R}^2$, consisting of $n$ pursuers and one evader. For $i \in [n]:=\{1,\ldots,n\}$, the state of the pursuer $i$ is its instantaneous position denoted by $p_i=(x_i,y_i)$. Similarly, the state of the evader is its instantaneous position denoted by $e=(x_e,y_e)$. Motivated by models studied in beyond-visual-range scenarios \cite{ref_model_garcia2021beyond}, each pursuer is assumed to move at its maximum speed, denoted by $v_i$, and its control input is given by its heading angle, denoted by $\theta_i$. 
Meanwhile, the evader's control input constitutes its speed, denoted by $\mu$, and its heading angle, denoted by $\psi$. In particular, the states of pursuer $i$ and the evader evolve as
\begin{align}
    &\ \dot{x}_i = v_i \cos \theta_i,  \; 
\dot{y}_i = v_i \sin \theta_i, \label{eq:P_dyn}
\\ &\		\dot{x}_e =\mu \cos \psi, \;
		\dot{y}_e =\mu \sin \psi, \label{eq:E_dyn}
\end{align}
where $\theta_i, \psi \in [0,2\pi)$, $\mu \in [0,\mu_m]$, $v_i > 0$, and $\mu_m \geq 0$.

We now present some key definitions and assumptions relevant to our analysis. The convex hull formed by $n$ pursuers in $\mathbb{R}^2$ is defined as
\begin{align}
& \mathcal{H}(p_1,p_2,\ldots,p_n) \nonumber
\\ & \qquad :=\big\{p \in \mathbb{R}^2 \mid p=\sum_{i=1}^{n} \lambda_i p_i, \lambda_i \geq 0, \sum_{i=1}^{n} \lambda_i=1\big\}. \label{eqn_convhull}
\end{align}

\begin{definition}[Encirclement Condition]
The evader is said to be encircled by the pursuers at time $t$ if it lies in the convex hull formed by the pursuers, i.e., $e(t) \in \mathcal{H}(p_1(t), \dots, p_n(t)) $. 
\end{definition}

For better readability, we denote $\mathcal{H}(p_1(t),\dots, p_n(t))$ by $\mathcal{H}$, whenever arguments are clear from the context.

\begin{definition}[Finite-Time Capture Condition] 
The evader is said to be captured in finite time if there exists at least one pursuer whose distance from the evader equals the capture radius\footnote{All pursuers are assumed to have the same capture radius $r_c > 0$.} $r_c$ at some finite time denoted by $t_c$, i.e., $d_{\min}(t)=\min_{i \in [n]} \,\|\, p_i(t) - e(t) \,\|_2 = r_c$ for some $t = t_c < \infty$.
\end{definition}

The notion of finite-time capture is meaningful because we assume that the interaction terminates when the evader is captured for the first time, possibly due to the evader being physically destroyed. Without this assumption, while the pursuer which first captures the evader can guarantee sustained capture at all future times (e.g., via pure pursuit with equal speed as the evader), analyzing
whether the evader remains encircled for $t > t_c$ remains an open problem.

\begin{definition}[Redundant Pursuers]
Pursuer $i$, where $i \in [n]$, is said to be redundant if the convex hull formed by all $n$ pursuers is identical to the convex hull formed by the remaining $n-1$ pursuers after removing its position $p_i(t)$.
\end{definition}

Since the presence of redundant pursuers does not influence our theoretical analysis, the following assumption is imposed to exclude them from further consideration.

\begin{assumption} \label{ass:redudant}
All $n$ pursuers are non-redundant. 
\end{assumption}

The notion of encirclement within a convex hull in the $\mathbb{R}^2$ space is meaningful only when there are at least three pursuers. Thus, we assume $n\geq3$.
The pursuit-evasion interaction terminates with a win for the pursuers if the finite-time capture condition holds and the encirclement condition holds at all times $t \leq t_c$. 
Conversely, the evader wins if it violates either of them. We now state the problem we investigate in this letter.
\begin{problem}\label{problem:stat}
Given an initial configuration of $n \geq 3$ non-redundant pursuers and one evader satisfying the encirclement condition at $t=0$, the goal is to design decentralized pursuer strategies that maintain encirclement at all times, \emph{and} guarantee capture in finite time, irrespective of and without knowledge of the speed and heading angle chosen by the evader.   
\end{problem}

The assumption that the pursuers are initialized in an encirclement configuration is standard in the existing literature addressing unknown evader policies with rigorous guarantees on either encirclement or capture \cite{ref_ratnoo_kumar2025cooperative,ref_slow_RMPC_encirclement_wang2025distributed,ref_early_bopardikar2014k,ref_RMPC_patra20253d}. The fundamental challenge lies in the fact that the instantaneous speed and heading angle of the evader are unknown to all the pursuers. However, we impose the following assumption on the maximum speed of the evader.

\begin{assumption}\label{assumption:evader_bound}
The maximum speed of the evader, denoted by $\mu_m$, is assumed to be known to the pursuers.
\end{assumption}

\noindent Note that the evader need not always move with speed $\mu_m$. Our analysis begins with an area-based approach to encirclement, which is presented in the next section.

\section{An Area-based approach to Encirclement} \label{section:area}

Consider $n$ pursuers and an evader such that the evader's initial position lies inside the convex hull formed by the pursuers' initial positions. The convex hull $\mathcal{H}$ can be triangulated by connecting the evader to consecutive pursuers, forming triangles with vertices $e$, $p_j$, and $p_k$ denoted by $\triangle ep_jp_k$, where $j,k \in [n]$ are consecutive integers modulo $n$.
The area of $\triangle ep_jp_k$ is denoted compactly by $A_{jk} := A(\triangle ep_jp_k)$. Figure \ref{fig:triangulation} illustrates two examples of triangulation with three and eight pursuers. We adopt the following conventions to ensure a consistent treatment of the vertices and areas of the triangles used throughout the letter.
\begin{enumerate}
    \item We fix a counterclockwise ordering for the vertices of each sub-triangle $\triangle ep_jp_k$
    starting from the evader's position when $e \in \operatorname{int}(\mathcal{H})$. For example, $\triangle ep_jp_k$, where $j,k \in [n]$ are consecutive integers modulo $n$, indicates that the vertices of the triangle are ordered from $e$ to $p_j$ to $p_k$ in a counterclockwise direction while $e \in \operatorname{int}(\mathcal{H})$. 
    This ordering is retained even if the evader exits $\mathcal{H}$, thus ensuring that $A_{jk}>0$ when $e \in \operatorname{int}(\mathcal{H})$, $A_{jk}=0$ when it resides on an edge and $A_{jk}<0$ when it exits $\mathcal{H}$ reflecting change in orientation to clockwise. This convention allows the area to serve as a quantitative indicator of containment within the convex hull.
    \item  By slight abuse of notation, $p_j$ denotes both the pursuer $j$ and its position depending on the context. 
\end{enumerate} 

\begin{figure}[t!]    
      \centering
      \begin{subfigure}{0.23\textwidth}
       \centering 
         \includegraphics[width=\textwidth]{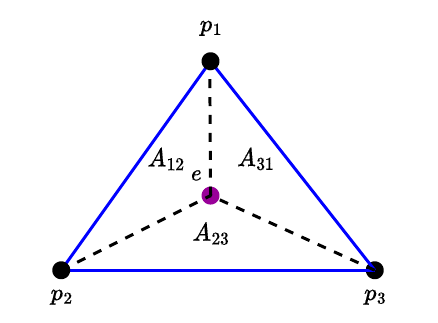}   
         \label{fig:EGPY}
      \end{subfigure}
     \begin{subfigure}{0.23\textwidth}
       \centering 
     \includegraphics[width=\textwidth]{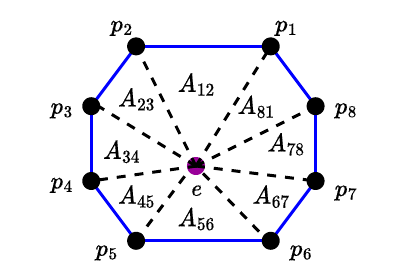}
      \label{fig:EGPN}
      \end{subfigure}
       \centering
      \caption{Examples of triangulation with an evader lying inside the convex hull of $3$ pursuers (left) and $8$ pursuers (right).}
      \label{fig:triangulation}
   \end{figure} 

The area $A_{jk}$ formed by the vertices $e=(x_e,y_e)$, $p_j=(x_j,y_j)$, and $p_k=(x_k,y_k)$ taken in appropriate order as per the above convention, can be computed as
\begin{equation} \label{eq:A_jk}
A_{jk} = \tfrac{1}{2}\bigl[x_e(y_j-y_k)+x_j(y_k-y_e)+x_k(y_e-y_j)\bigr].    
\end{equation}
The area $A_{jk}$ is a function of the instantaneous position of the agents, namely $e(t)$, $p_j(t)$, and $p_k(t)$, which evolve over time. 
We compute the expression for its time derivative, which can be decomposed into pursuer and evader terms as
\begin{align}
2\dot{A}_{jk} & =T_p+T_e, \quad \text{where} \label{eq:Ajkdot_decompose}
\\ T_p & =v_j \cos \theta_j (y_k-y_e) + v_k \cos \theta_k (y_e-y_j) \nonumber
\\ & \quad + v_j \sin \theta_j (x_e-x_k) + v_k \sin \theta_k (x_j-x_e), \label{eq:Tp}
\\ & =d_{ek} v_j \cos(\theta_j-\beta)+d_{ej} v_k \cos(\theta_k-\gamma), \label{eq:Tp_compact}
 \\ T_e & =\mu \cos\psi (y_j-y_k) + \mu \sin\psi (x_k-x_j). \label{eq:Te}
 \\ &= -\mu d_{jk} \sin(\alpha-\psi), \quad \text{where} \label{eq:Te_compact}
\\ \quad d_{ej} & = \|e-p_j\|, \, d_{ek}=\|e-p_k\|, \, \label{eq:dej_ek_jk}
\\
\beta &=\tan^{-1}\!\left(\frac{x_e - x_k}{\,y_k - y_e}\right), \, \, \gamma =\tan^{-1}\!\left(\frac{x_j - x_e}{\,y_e - y_j}\right), \nonumber 
\\ d_{jk} & =\|p_j-p_k\|, \quad \alpha =\tan^{-1}\!\left(\frac{y_k-y_j}{\, x_k-x_j}\right). \label{eq:beta_gamma_delta}
\end{align}
We now define the notion of an active edge and active pursuers.   

\begin{definition} [Active Edge and Active Pursuers]
An edge $p_jp_k$, where $j,k \in [n]$ are consecutive integers modulo $n$, is said to be \emph{active} at time $t$ if there exists $\lambda \in [0,1]$ such that $e(t)=\lambda p_j(t)+ (1-\lambda) p_k(t)$. The corresponding pursuers $j$ and $k$ are called \emph{active pursuers} at time $t$.
\end{definition}

In other words, if the evader lies on the line joining the positions $p_j(t)$ and $p_k(t)$, then the edge $p_jp_k$ is said to be active. The endpoints $p_j$ and $p_k$ are named as per the convention of counterclockwise vertex ordering starting from the evader's position while it was in the interior of the convex hull.\footnote{Note that, $p_j$ (or $p_k$) does not refer to (the position of) a specific pursuer throughout pursuit. According to conventions adopted for the area and the vertices, whenever an edge is active, $p_j$ and $p_k$ denote the pursuers to the right and left of the evader, respectively, defined with respect to the evader’s outward heading from the convex hull.} Under the adopted convention, we know that $A_{jk}\geq0$ when $e\in\mathcal{H}$, and $A_{jk}<0$ if the evader exits through that edge. Hence, whenever $A_{jk}=0$, encirclement is maintained if and only if $\dot{A}_{jk}\ge 0$. The following proposition summarizes the above discussion.

\begin{proposition} \label{proposition:evader_str}
Following the conventions set up for area $A_{jk}$, encirclement is guaranteed if and only if
\begin{enumerate}
\item either the evader lies in the interior of the convex hull at all times, or
\item whenever an edge $p_j p_k$ is active, then there exists strategies for the active pursuers $j$ and $k$ to ensure that $\dot{A}_{jk}\geq 0$ for all possible values of $\mu \leq \mu_m$ and $\psi \in  [0, 2\pi)$.
\end{enumerate}
\end{proposition}

It is clear from Proposition \ref{proposition:evader_str} that the evader can only escape at time $t$ if an edge $p_jp_k$ is active, that is, $A_{jk}=0$ \emph{and} the strategies of the active pursuers and the evader are such that $\dot{A}_{jk} < 0$. We also know that when the evader is on the edge formed by active pursuers $p_j$ and $p_k$, the points $e$, $p_j$, and $p_k$ are collinear. Let $u_{jk}$ be the unit vector pointing from $p_j$ towards $p_k$, denoted by
\begin{equation}\label{enc_eq:u}
u_{jk}=\frac{p_k - p_j}{\|p_k - p_j\|} = \begin{bmatrix}
    \cos \alpha \\ \sin \alpha
\end{bmatrix},
\end{equation}
where $\alpha$ is given from \eqref{eq:beta_gamma_delta}.
When the edge $p_jp_k$ is active, the vectors from the active pursuers $p_j$ and $p_k$ to the evader can thus be expressed in terms of $u_{jk}$ as follows:
\begin{equation}
\begin{aligned}
& e - p_j = u_{jk} \|e - p_j\|, \;
e - p_k = -u_{jk} \|e- p_k\|.     
\end{aligned}  
\label{eq:enc_pjpk_u}
\end{equation}
Based on the above discussion, we now present a class of strategies for the pursuers that guarantee encirclement of the evader regardless of its chosen speed and heading angle. Specifically, the proposed heading angles of the pursuers depend on the current positions of the evader and the (active) pursuers.

\begin{theorem} \label{theorem:enc}
Let the evader's initial position satisfy $e(0) \in \mathcal{H}(0)$. Let each pursuer $i \in [n]$ choose any admissible heading $\theta_i\in [0,2\pi)$ when none of the edges are active. Whenever an edge $p_jp_k$ becomes active, the pursuers constituting the active edge switch to the following strategy:
\begin{equation} \label{eq:enc_edge}
        \begin{aligned}
            \begin{bmatrix}
           \cos \theta_j \\[0.5ex]
            \sin \theta_j
          \end{bmatrix} =R_{\varphi_j}\,\frac{e - p_j}{\|e - p_j\|} 
, \\
             \begin{bmatrix}
  \cos \theta_k \\[0.5ex]
  \sin \theta_k
\end{bmatrix} = R_{\varphi_k}^{\top}\,\frac{e - p_k}{\|e - p_k\|},
        \end{aligned}
    \end{equation}   
    \begin{equation} \label{eqn:enc_R_theta} \text{where} \quad
        R_{\varphi_{l \in \{j,k\}}} = 
        \begin{bmatrix}
            \cos \varphi_l & \sin \varphi_l \\
            -\sin \varphi_l & \cos \varphi_l
        \end{bmatrix},
    \end{equation}
with the control parameter $\varphi_j$ (respectively, $\varphi_k$) being the angle of rotation with respect to the vector $e-p_j$ (respectively, $e-p_k$), considered as positive in the outward direction from $\mathcal{H}$.
This strategy guarantees encirclement at all times, regardless of the speed and heading angle chosen by the evader, if and only if the angles $\varphi_j$ and $\varphi_k$ satisfy
\begin{equation} \label{eq:necess_suff_enc}
v_j d_{ek} \sin \varphi_j + v_k d_{ej} \sin \varphi_k \geq d_{jk} \mu_m. 
\end{equation}
\end{theorem}

\begin{proof}
Under the proposed strategy when edge $p_jp_k$ is active, substituting \eqref{eq:enc_pjpk_u} into \eqref{eq:enc_edge} yields
\begin{equation} \label{enc_eq:pjpk_dynamics_u}
\begin{aligned}
\dot{p}_j &= v_j \begin{bmatrix}
    \cos \theta_j \\ \sin \theta_j
\end{bmatrix} =v_j\, R_{\varphi_j   }\, u_{jk} = v_j \begin{bmatrix}
    \cos (\alpha-\varphi_j) \\ \sin (\alpha-\varphi_j)
\end{bmatrix}
, \\  \dot{p}_k &= v_k \begin{bmatrix}
    \cos \theta_k \\ \sin \theta_k
\end{bmatrix} = -v_k\, R_{\varphi_k}^{\top}\, u_{jk} = -v_k \begin{bmatrix}
    \cos (\alpha+\varphi_k) \\ \sin (\alpha+\varphi_k)
\end{bmatrix}. 
\end{aligned}
\end{equation}

For a general vector $v=[v_x,v_y]^{\top}$, let $v^\perp:=[-v_y,v_x]^{\top}$, which is orthogonal to $v$. Then, $T_p$ from \eqref{eq:Tp} is written as 
 \begin{align} 
  T_p & = \dot{p}^{\top}_j  (e-p_k)^{\perp}  + \dot{p}^{\top}_k  (p_j-e)^{\perp} \nonumber
  \\ &= -\|e-p_k\| (v_j R_{\varphi_j}u_{jk})^{\top} u_{jk}^{\perp} + \nonumber \\ & \quad \quad \|p_j-e\| (v_k R_{\varphi_k}^{\top}u_{jk})^{\top} u_{jk}^{\perp}, \label{enc_eq:T_p_second}
  \\ & =-v_j \|e-p_k\| (-\sin \varphi_j) + v_k \|p_j-e\| (\sin \varphi_k), \nonumber 
  \\ & =v_j d_{ek} \sin \varphi_j + v_k d_{ej} \sin \varphi_k, \label{enc_eq:Tp_final}
 \end{align}  
 where the third equality can be obtained from \eqref{enc_eq:T_p_second} by using \eqref{enc_eq:pjpk_dynamics_u} and \eqref{enc_eq:u}.

\smallskip

Adding $T_p$ from \eqref{enc_eq:Tp_final} and $T_e$ obtained in \eqref{eq:Te_compact} \footnote{The evader term $T_e$ does not depend upon pursuer strategies.} yields
\begin{equation}\label{enc_eq:final_Ajkdot_T2}
2\dot{A}_{jk}=v_j d_{ek} \sin \varphi_j+ v_k d_{ej} \sin \varphi_k-d_{jk}\mu \sin(\alpha-\psi). 
\end{equation}

From Proposition \ref{proposition:evader_str}, we know that encirclement is guaranteed if and only if $\dot{A}_{jk} \geq 0$, or equivalently, 
\begin{equation} \label{enc_eq:varphi_range}
v_j d_{ek} \sin \varphi_j + v_k d_{ej} \sin \varphi_k \geq d_{jk} \mu \sin(\alpha-\psi) 
\end{equation}
for all $\psi \in [0, 2\pi)$ and all $\mu \in [0,\mu_m]$. In other words, the pursuers must ensure that $v_j d_{ek} \sin \varphi_j + v_k d_{ej} \sin \varphi_k \geq d_{jk} \mu_m$, since $\mu_m \geq \mu$ and $\sin(\alpha-\psi) \in [-1,1]$ for any choice of $\psi$. For any value of $\varphi_j$ and $\varphi_k$ not satisfying this condition, there always exists an evader strategy (a set of speed and heading) that can make the closed-loop expression of $\dot{A}_{jk}$ from \eqref{enc_eq:final_Ajkdot_T2} negative, leading to violation of the encirclement condition. This completes the proof.
\end{proof}
\begin{figure}[t!]
\centering
\vspace{8pt}
\includegraphics[width=0.4\textwidth]{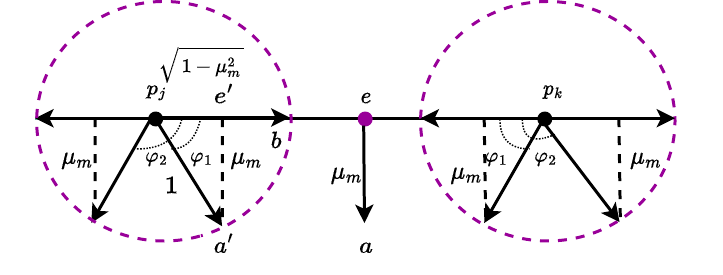}
\caption{Range of strategies for the active pursuers that guarantee encirclement of the evader, regardless of its strategy, when $v_i=1$ for all $i \in [n]$ and $\mu_m \leq 1$. In this figure, the angles $\varphi_1$ and $\varphi_2$ are given by $\varphi_1=\sin^{-1}(\mu_m)$ and $\varphi_2=\pi-\sin^{-1}(\mu_m)$. Thus, $\varphi_l \in [\varphi_1,\varphi_2]$, where $l \in \{j,k\}$ and the angles are measured from the direction of the respective unit vectors towards the evader for $p_j$ and $p_k$ and considered as positive in the outward direction from $\mathcal{H}$.} 
\label{fig:edge_strategy}
\end{figure}

A key result involving the pursuers' speed and the angles $\varphi_j$ and $\varphi_k$, which establish sufficient conditions for encirclement, is summarized in Corollary \ref{corollary:range}. 

\begin{corollary} \label{corollary:range}
Suppose $v_i=1$ for all $i \in [n]$ and $\mu_m \leq 1$. Let active pursuers choose the angles $\varphi_j$ and $\varphi_k$ from $\varphi_l \in [\sin^{-1}(\mu_m), \pi-\sin^{-1}(\mu_m)]$, for $l \in \{j,k\}$. Then, \eqref{eq:necess_suff_enc} is satisfied, and encirclement is guaranteed at all times, irrespective of the strategy chosen by the evader. 
\end{corollary}

\begin{proof}
When any edge $p_jp_k$ becomes active, the active pursuers with $v_j=v_k=1$ choose $\varphi_l \in [\sin^{-1}(\mu_m), \pi-\sin^{-1}(\mu_m)]$, for $l \in \{j,k\}$, we have $\sin \varphi_j \geq \mu_m$, and $\sin \varphi_k \geq \mu_m$. Therefore,   
\begin{equation*}
    v_j d_{ek} \sin \varphi_j + v_k d_{ej} \sin \varphi_k \geq (d_{ek} + d_{ej}) \mu_m = d_{jk} \mu_m,
\end{equation*}
i.e., \eqref{eq:necess_suff_enc} is satisfied at all times. Thus, encirclement is guaranteed at all times regardless of the evader's strategy. 
\end{proof}

Figure \ref{fig:edge_strategy} provides a schematic illustration of the kinematics associated with the proposed range of strategies in Corollary \ref{corollary:range} that guarantee encirclement. Intuitively, when the evader activates an edge defined by consecutive pursuers $p_j$ and $p_k$ leading to $A_{jk} = 0$, pursuers $p_j$ and $p_k$ move outward from the convex hull at a rate sufficient to ensure that even the worst-case evader motion keeps the evader confined within the hull. Therefore, the evader is unable to escape the convex hull $\mathcal{H}$ irrespective of its chosen control input.

\begin{remark}
The pursuer strategies described in Corollary \ref{corollary:range} are decentralized since each active pursuer can independently select the angles $\varphi_j$ and $\varphi_k$ without communicating with any other pursuer. It only needs to know the position of the other pursuers and the evader. Furthermore, the encirclement guarantee holds irrespective of the evader’s position when the edge $p_jp_k$ becomes active, i.e., independent of the values of $d_{ej}$ and $d_{ek}$. In contrast to the condition \eqref{eq:necess_suff_enc} established in Theorem \ref{theorem:enc}, the sufficient condition in Corollary \ref{corollary:range} can be verified a priori, before the pursuit begins.
\end{remark}

\begin{remark}
It follows from Theorem \ref{theorem:enc} and Corollary \ref{corollary:range} that the encirclement guarantee holds even when the value of $\mu_m$ is unknown to the pursuers, as long as we have homogeneous pursuers not slower than the evader, that is, $v_j=v_k=1$ and $\mu_m \leq 1$. In that case, the pursuers would move normally outward from the active edge at maximum speed normalized to unity. This is evident by substituting $v_j=v_k=1$ and $\mu_m=1$ in \eqref{eq:necess_suff_enc} leading to $d_{ek} \sin \varphi_j + d_{ej} \sin \varphi_k \geq d_{jk}$, which holds only when $\varphi_j=\varphi_k=\dfrac{\pi}{2}$.
\end{remark}

\begin{remark}
It can be observed from the expression of $\dot{A}_{jk}$ from \eqref{enc_eq:final_Ajkdot_T2} that $\varphi_j^\star=\varphi_k^\star=\dfrac{\pi}{2}$ and $\psi^\star=\alpha-\dfrac{\pi}{2}$ are the optimal solutions to $\max_{\varphi_j^\star, \varphi_k^\star}   \min_{\psi^\star} \dot{A}_{jk}(\psi, \varphi_j, \varphi_k)$, which corresponds to the strategies to move normally outwards from $\mathcal{H}$ for both the pursuers and the evader.   
\end{remark}

\section{Encirclement Guaranteed Capture Strategy for the Pursuers} \label{section:proposed_approach}

In the previous section, we characterized pursuer strategies that guarantee encirclement. We now focus on deriving finite-time capture guarantees. To this end, we assume $v_i=1$ for $i \in [n]$, and $\mu_m<1$, that is, all the pursuers move with unit speed, and the evader is strictly slower than the pursuers. Under these assumptions, we now show that a subset of strategies from the strategy set identified in Corollary~\ref{corollary:range} is sufficient to guarantee finite-time capture of the evader. 

\begin{theorem}\label{theorem:enc_capture}
Suppose $v_i=1$ for $i \in [n]$, $\mu_m<1$, and the evader is initially encircled, that is, $e(0) \in \mathcal{H}(0)$. Then, the following pursuit strategy is sufficient to ensure that the evader remains encircled throughout pursuit and is captured in finite time regardless of the strategy it adopts. 
\begin{enumerate}[leftmargin=*, itemsep=1ex, label=\textbf{\arabic*.}]
    \item \textbf{Interior Phase:} 
    When the evader $e(t) \in \operatorname{int}(\mathcal{H}(t))$, all pursuers follow the pure pursuit strategy given by
    \begin{equation} \label{eq:pure_pursuit}
        \dot{p}_i = \begin{bmatrix}
        \cos \theta_i \\[0.5ex]
        \sin \theta_i
    \end{bmatrix} = \frac{e - p_i}{\|e - p_i\|},
        \qquad \forall i \in [n].
    \end{equation}
    \item \textbf{Edge Phase:} 
When any edge $p_jp_k$ becomes active, the pursuers constituting the active edge switch to the strategy proposed in \eqref{eq:enc_edge} with any $\varphi_l \in [\sin^{-1}(\mu_m),\frac{\pi}{2}-\sin^{-1}(1-\mu_m)]$, where $l \in \{j,k\}$, while the inactive pursuers continue to follow pure pursuit. 
\end{enumerate}
\end{theorem}
\begin{proof}
Note that the proposed edge-phase angle range $\varphi_l \in [\sin^{-1}(\mu_m),\frac{\pi}{2}-\sin^{-1}(1-\mu_m)]$, where $l \in \{j,k\}$ and $\mu_m \in [0,1)$, is a subset of the admissible range  $\varphi_l \in [\sin^{-1}(\mu_m),\pi-\sin^{-1}(\mu_m)]$, for which encirclement is guaranteed at all times following Corollary \ref{corollary:range}.

We prove finite-time capture of the evader with the proposed pursuers' strategy using a Lyapunov-based analysis. We select the following candidate Lyapunov function
\begin{equation}\label{eq:lyapunov}
V = \sum_{i=1}^{n} \|p_i - e\| = \sum_{i=1}^{n} d_i,
\end{equation}
which is positive definite and represents the sum of distances between the pursuers and the evader. We now analyze $\dot{V}$ for different phases of pursuit assuming that $\|p_i-e\| > r_c$ for all the pursuers, that is, the evader is not yet captured.

\noindent \textbf{a) Interior Phase.} When $e \in \operatorname{int}(\mathcal{H}(t))$, all pursuers use pure pursuit. We then compute the time derivative of $V$ as
\begin{align}
    \dot{V} &= \sum_{i=1}^{n} \dot{d}_i= \sum_{i=1}^{n} \frac{(p_i-e)^\top \dot{p}_i}{\|e - p_i\|} + \frac{(e-p_i)^\top \dot{e}}{\|e - p_i\|}  \label{eq:vdot_general} \\
	&= \sum_{i=1}^{n} \left(\frac{(p_i-e)^\top}{\|e - p_i\|}\right) \left(\frac{e-p_i}{\|e - p_i\|}\right) + \frac{(e-p_i)^\top \dot{e}}{\|e - p_i\|} \nonumber \\
	&= \sum_{i=1}^{n} -1 + \frac{(e-p_i)^\top \dot{e}}{\|e - p_i\|} \leq -n(1-\mu_m),
	\label{eq:vdot_interior}
\end{align}
where the inequality follows from the Cauchy-Schwarz inequality, under which
\begin{align*}
     & |(e - p_i)^T \dot{e}| \leq \|e - p_i\| \|\dot{e}\| = \mu \|e - p_i\|
    \\  \implies & \frac{(e-p_i)^\top \dot{e}}{\|e - p_i\|}-1 \leq \mu-1 \leq \mu_m-1. 
\end{align*}
Since $\mu_m < 1$, the Lyapunov rate $\dot{V} \leq -n(1-\mu_m)$ is bounded by a finite negative quantity. 

In our setting, the capture takes place when at least one pursuer is within distance $r_c$ from the evader, that is, at a finite positive value of the Lyapunov function. Thus, the constant rate of decay of the Lyapunov function results in finite-time capture. In particular, the evader is captured in finite time if it stays in the interior of the convex hull $\mathcal{H}(t)$ throughout the pursuit.

\noindent {\bf b) Edge Phase.} Suppose an edge $p_jp_k$ is active at time $t$. In this case, the Lyapunov rate can be expressed as
\begin{equation}
\dot{V}= \dot{d}_j + \dot{d}_k + \sum^n_{i=1, i \notin \{j,k\}} \dot{d}_i,   
\end{equation} 

Since pursuers that are not active adopt pure pursuit, following the earlier analysis, we have 
\begin{equation}\label{eq:Lyapunov_inactive}
\dot{d}_i \leq -(1-\mu) \leq -(1-\mu_m),  \quad \text{for} \quad i \neq j, k. 
\end{equation}
 
We now consider the active pursuers. 
We compute $\dot{d}_j$ by substituting the expression of $\dot{p}_j$ from \eqref{eq:enc_edge} with $v_j=1$ as
\begin{align}
    \dot{d}_j & = \frac{(p_j-e)^\top \dot{p}_j}{\|e - p_j\|} + \frac{(e-p_j)^\top \dot{e}}{\|e - p_j\|}  \nonumber \\
	&= \left(\frac{(p_j-e)^\top}{\|e - p_j\|}\right) \left( R_{\varphi_j}\, \frac{e-p_j}{\|e - p_j\|}\right) + \frac{(e-p_j)^\top \dot{e}}{\|e - p_j\|} \nonumber \\
	&= -\cos \varphi_j + \frac{(e-p_j)^\top \dot{e}}{\|e - p_j\|}. \label{eq:Vj_dot}
\end{align}
On similar lines, we compute
\begin{equation} \label{eq:Vk_dot}
\dot{d}_k= -\cos \varphi_k + \frac{(e-p_k)^\top \dot{e}}{\|e - p_k\|}.
\end{equation}
Adding \eqref{eq:Vj_dot} and \eqref{eq:Vk_dot}, we obtain
\begin{align}
    \dot{d}_j+\dot{d}_k & = -\cos \varphi_j-\cos \varphi_k+ \frac{(e-p_j)^\top \dot{e}}{\|e - p_j\|}+\frac{(e-p_k)^\top \dot{e}}{\|e - p_k\|}  \nonumber
    \\ & = -\cos \varphi_j-\cos \varphi_k, \label{eq:vjdot_vkdot_sum_final}
\end{align} 
where the last equality follows because from \eqref{eq:enc_pjpk_u}, we have
\begin{equation}
\begin{aligned}
& \frac{(e - p_j)^\top}{\|e - p_j\|} = u^{\top}_{jk} , \;
\frac{(e - p_k)^\top}{\|e - p_k\|} = -u^{\top}_{jk}.     
\end{aligned}  
\end{equation}

Combining the terms for the inactive and active pursuers from \eqref{eq:Lyapunov_inactive} and \eqref{eq:vjdot_vkdot_sum_final}, we obtain
\begin{align}
\dot{V} & \leq -(n-2)(1-\mu_m)-\cos \varphi_j-\cos \varphi_k, \nonumber
\\ & \leq -n(1-\mu_m), 
\label{eq:vdot_common_rate}
\end{align}
as before; the last inequality follows because $\varphi_l \in [\sin^{-1}(\mu_m),\frac{\pi}{2}-\sin^{-1}(1-\mu_m)]$ implies $1-\mu_m \leq \cos \varphi_l \leq \sqrt{1-\mu_m^2}$ for $l \in \{j,k\}$. Therefore, the finite-time convergence established in the interior phase continues to hold when the evader is on an edge, and also when the evader potentially switches at an arbitrary frequency between the interior and edge phases, since we have found a common bound that applies to both phases. 

Integrating \eqref{eq:vdot_common_rate} from $t=0$ to $t_c$, we obtain
\begin{align}
    t_c & \leq \frac{V_0-V_c}{n(1-\mu_m)} \leq \frac{V_0-nr_c}{n(1-\mu_m)}, \label{eq:common_bound}
\end{align}
where $V_0$ is the value of the Lyapunov function at $t=0$ and $V_c$ is the value of the Lyapunov function at the time of capture $t_c$. 
Note that capture occurs at the first time instant when $d_i = r_c$ for some $i$, which implies that throughout the game $d_i=\|p_i-e\|\geq r_c$ for all $i \in [n]$, with equality holding at the time of capture if the pursuer $i$ captures the evader. This implies $V_c \geq nr_c$, with equality when all the pursuers simultaneously reach a distance $r_c$ from the evader. 
\end{proof}

The strategy of active pursuers in Theorem \ref{theorem:enc_capture} switches between two control laws at the surface $A_{jk}=0$ (for any $j,k \in [n]$), making the right-hand side of the pursuer dynamics discontinuous. We now prove the existence of solutions for the resulting discontinuous ordinary differential equation (ODE) under mild assumptions on the strategy of the evader. Let $ z = (p_1, \ldots, p_n, e)^{\top} \in \mathbb{R}^{2(n+1)}$ be the augmented vector consisting of all agent positions stacked together. The closed-loop dynamics under the strategy proposed in Theorem \ref{theorem:enc_capture} can be expressed as a single ODE as
\begin{equation}\label{eq:switch}
\dot z=f(z)=
\begin{cases}
f_{\mathrm{int}}(z),
& \text{if} \, A_{jk}(z)>0,\ \forall\, j,k\in[n],\\[4pt]
f^{jk}_{\mathrm{edge}}(z),
&  \text{if} \; \exists\, j,k\in[n]\ \text{s.t.}\ A_{jk}(z)=0,
\end{cases}
\end{equation}
where
\begin{align}
 & f_{\mathrm{int}}(z) = \left(
    f_1,\ldots,f_n,\dot{e}
  \right)^\top \in \mathbb{R}^{2(n+1)}, \label{eq:fint}\\
 & f^{jk}_{\mathrm{edge}}(z) = \left(f_1,\ldots,f^{\mathrm{edge}}_j,f^{\mathrm{edge}}_k,\ldots,f_n,
    \dot{e}
  \right)^\top \in \mathbb{R}^{2(n+1)}, \label{eq:fedge} \\
 &f_i  =  \frac{e-p_i}{\|e-p_i\|} \quad \forall i \in [n], \\
 &f^{\mathrm{edge}}_j =  R_{\varphi_j}\frac{e-p_j}{\|e-p_j\|}, \; \; f^{\mathrm{edge}}_k= R_{\varphi_k}^{\top}\frac{e-p_k}{\|e-p_k\|},
\end{align}
where inactive pursuers $i \in [n]\setminus\{j,k\}$ follow pure pursuit and $R_{\varphi_{l \in \{j,k\}}}$ denotes the rotation matrix in \eqref{eqn:enc_R_theta}.

We now seek solutions in the sense of Filippov \cite{ref_cortes_cortes2008discontinuous}, wherein the discontinuous ODE~\eqref{eq:switch} is replaced by the differential inclusion $\dot{z}(t) \in \mathcal{F}[f](z(t))$, where $\mathcal{F}[f](z)$ denotes the closed convex hull of the limiting values of $f$ in a neighborhood of $z$, and we verify the existence of such solutions in the following proposition.

\begin{proposition}
\label{proposition:existence}
Consider the closed-loop system \eqref{eq:switch} with $v_i = 1$ for all $i \in [n]$, $\mu_m < 1$, $e(0) \in \mathcal{H}(0)$, the strategies of the pursuers are in accordance with Theorem $2$, and the strategy chosen by the evader is a continuous function of time, i.e., $\dot{e}$ is continuous. Then, for every initial condition $z_0 \in \mathcal{X} := \bigcap_{(j,k) \in \mathcal{N}_0} \{z : A_{jk}(z) \geq 0\}$, where $\mathcal{N}_0$ is the set of all consecutive integers modulo $n$ from $[n]$, there exists a Filippov solution $z : [0, \infty) \to \mathcal{X}$ of \eqref{eq:switch}.
\end{proposition}
\begin{proof}
We verify the two hypotheses of \cite[Proposition~3]{ref_cortes_cortes2008discontinuous}: measurability and local essential boundedness of $f$.
\noindent\textbf{(i) Measurability of $f$:}
When $\dot{e}$ is continuous, both $f_{\mathrm{int}}$ and $f^{jk}_{\mathrm{edge}}$ are individually continuous because each component is a smooth function of $z$ whenever $\|e-p_i\| > 0$ for all $i \in [n]$, which holds throughout pursuit since $\|e-p_i\| \geq r_c > 0$. However, $f$ in \eqref{eq:switch} is discontinuous on each switching surface $\mathcal{S}_{jk}:=\{z \in \mathbb{R}^{2(n+1)} : A_{jk}(z) = 0\}$, defined for $(j,k) \in \mathcal{N}_0$.
Hence, $f$ is continuous on $\mathcal{X} \setminus \mathcal{S}$, where $\mathcal{X} := \bigcap_{(j,k) \in \mathcal{N}_0} \{z : A_{jk}(z) \geq 0\}$ and $\mathcal{S} = \bigcup_{(j,k) \in \mathcal{N}_0} \mathcal{S}_{jk}$.
 Hence, $f$ is continuous on
$\mathcal{X} \setminus \mathcal{S}_{jk}$, where $\mathcal{X} := \{z : A_{jk}(z) \geq 0, \text{ for all consecutive } j, k \in [n]\}$.
Each switching surface $\mathcal{S}_{jk}=\{z \in \mathbb{R}^{2(n+1)} : A_{jk}(z) = 0\}$ is defined
by one smooth equation in $\mathbb{R}^{2(n+1)}$. Since
$\nabla_z A_{jk} \neq 0$ on $\mathcal{S}_{jk}$ whenever $p_j = p_k = e$ is not true (which holds throughout the pursuit since $r_c > 0$), \cite[Thm.~5-1]{ref_spivak2018calculus} guarantees that $\mathcal{S}_{jk}$ is a smooth, $2n+1$ dimensional hypersurface in $\mathbb{R}^{2(n+1)}$; hence, it has Lebesgue measure zero. Since there are $n$ such surfaces (one per each edge of the convex hull formed by $n$ pursuers), $\mathcal{S} = \bigcup_{(j,k) \in \mathcal{N}_0} \mathcal{S}_{jk}$ has Lebesgue measure zero because it is a finite union of measure-zero sets \cite{ref_real_royden2010real}. Thus, $f$ is continuous except on a set of Lebesgue measure zero. Therefore, it is measurable \cite{ref_real_royden2010real}.

\medskip

\noindent \noindent\textbf{(ii) Local essential boundedness of $f$:} Each pursuer moves at unit speed under both $f_{\mathrm{int}}$ and $f^{jk}_{\mathrm{edge}}$, since $\|f_i\| = 1$ for all $i \in [n] \setminus \{j,k\}$ and the matrices $R_{\varphi_l}$ for $l \in \{j,k\}$
are orthogonal ($R_{\varphi_l}^{\top}R_{\varphi_l} = I$), giving 
$\|f^{\mathrm{edge}}_l\|= 1$ for $l \in \{j,k\}$. Additionally, we have $\|\dot{e}\| \leq \mu_m$. This yields $\|f(z)\|^{2} \leq n + \mu_m^{2}$ uniformly over $\mathcal{X}$, so $f$ is locally essentially bounded. \\
Thus, both the conditions of \cite[Prop.~3]{ref_cortes_cortes2008discontinuous} are satisfied, and therefore, for every $z_0 \in \mathcal{X}$, there exists a Filippov solution of \eqref{eq:switch}.

Thus, both the conditions of \cite[Proposition~3]{ref_cortes_cortes2008discontinuous} are satisfied, and therefore, for every $z_0 \in \mathcal{X}$, there exists a Filippov solutions of \eqref{eq:switch}.
\end{proof}

We note that the continuity assumption on the evader’s strategy is made to facilitate the proof of existence of Filippov solutions. Relaxing this assumption remains open for future work. Nevertheless, numerical results, reported in the following section, show that the strategy proposed in Theorem 2 is effective against certain types of discontinuous evader strategies as well.

\begin{figure*}[t!]
      \centering
      \begin{subfigure}{0.24\textwidth}
       \centering 
         \includegraphics[width=\textwidth]{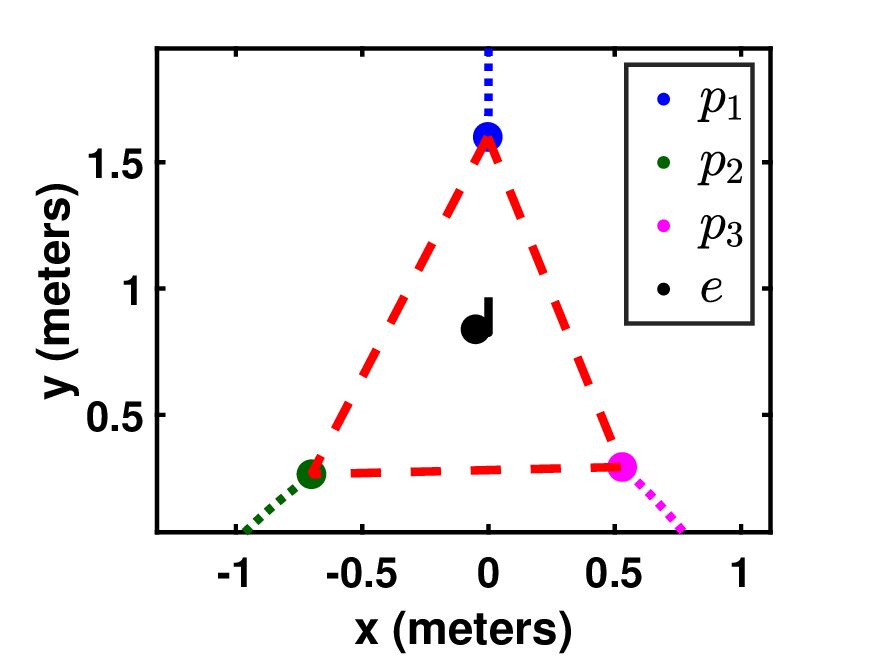}   
         \caption{$t=0.45 \, \text{s}$}
         \label{fig1:greedy_snap}
      \end{subfigure}
      \begin{subfigure}{0.24\textwidth}
       \centering 
     \includegraphics[width=\textwidth]{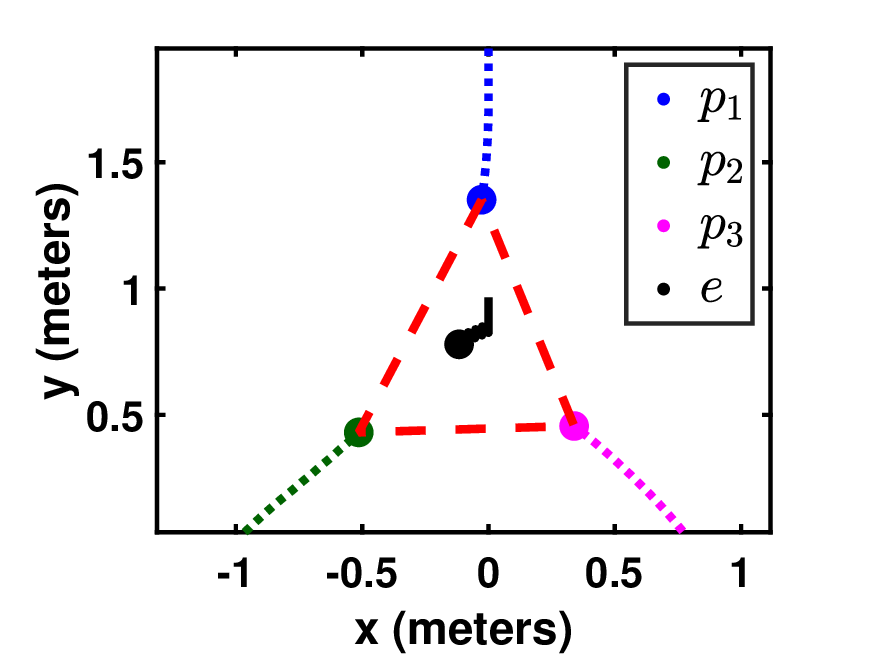}
     \caption{$t=0.7 \, \text{s}$}
      \label{fig2:greedy_snap}
      \end{subfigure}
    \begin{subfigure}{0.24\textwidth}
       \centering 
     \includegraphics[width=\textwidth]{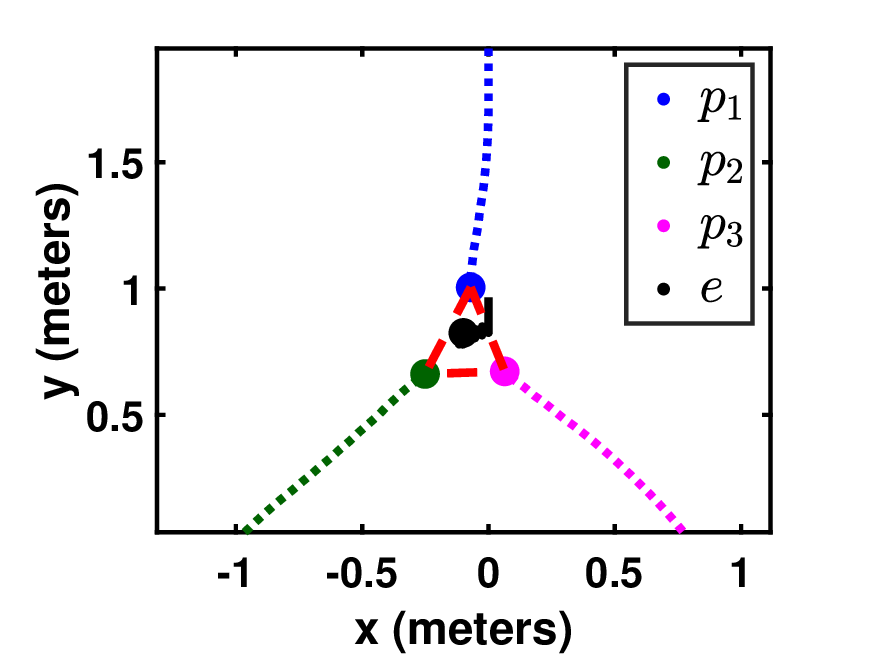}
     \caption{$t=1.05 \, \text{s}$}
      \label{fig3:greedy_snap}
      \end{subfigure}
       \begin{subfigure}{0.24\textwidth}
       \centering 
        \includegraphics[width=\textwidth]{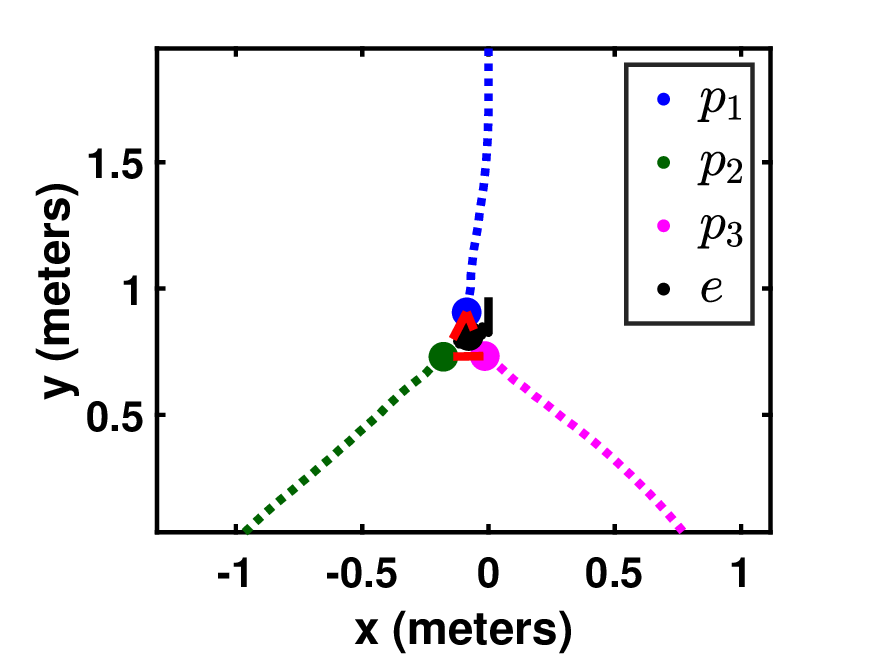}
         \caption{$t=1.25 \, \text{s}$}
          \label{fig4:greedy_snap}
          \end{subfigure}
       \centering
      \caption{Encirclement guaranteed capture under the {\bf greedy policy} of the evader with three pursuers. The evader moves away from the nearest pursuer at each instant throughout the pursuit. The evader is always encircled and is captured at $t=1.25 \, \text{s}$.}
      \label{fig:phase_greedy}
   \end{figure*}
\begin{figure*}[t!]
      \centering
      \begin{subfigure}{0.24\textwidth}
       \centering 
         \includegraphics[width=\textwidth]{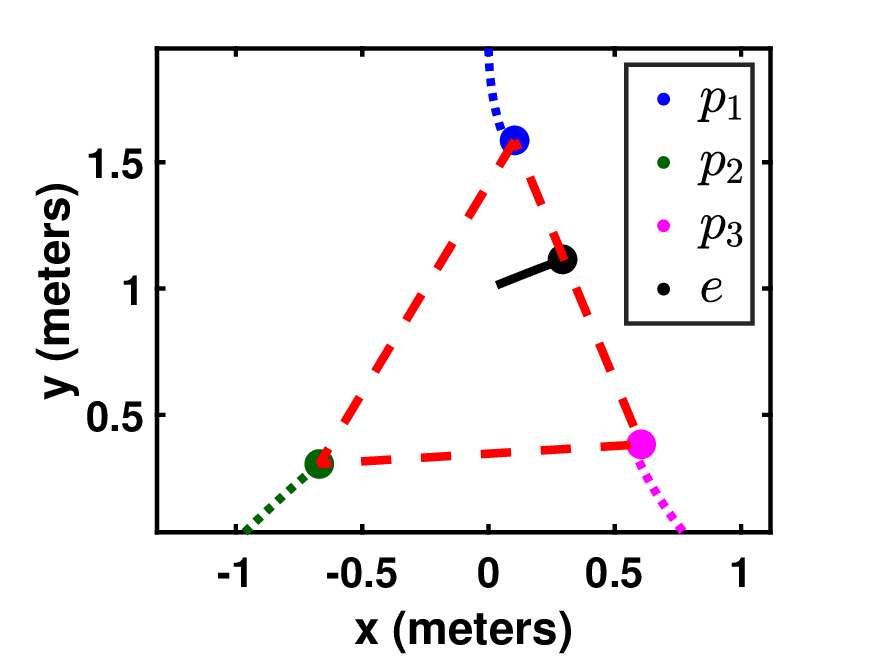}   
         \caption{$t=0.5 \, \text{s}$}
         \label{fig1:switch_snap}
      \end{subfigure}
    \begin{subfigure}{0.24\textwidth}
       \centering 
     \includegraphics[width=\textwidth]{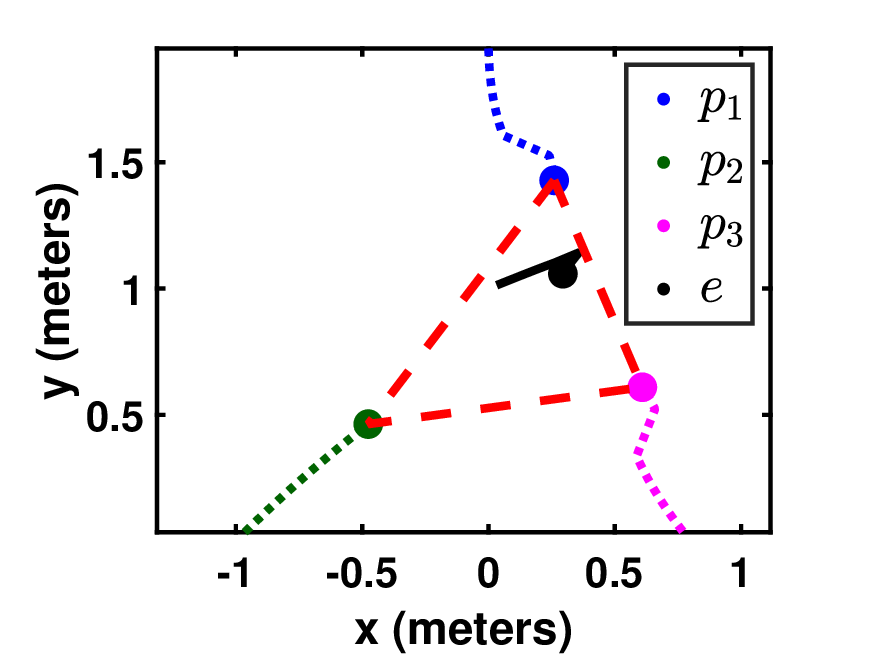}
     \caption{$t=0.75 \, \text{s}$}
      \label{fig2:switch_snap}
      \end{subfigure}
    \begin{subfigure}{0.24\textwidth}
       \centering 
        \includegraphics[width=\textwidth]{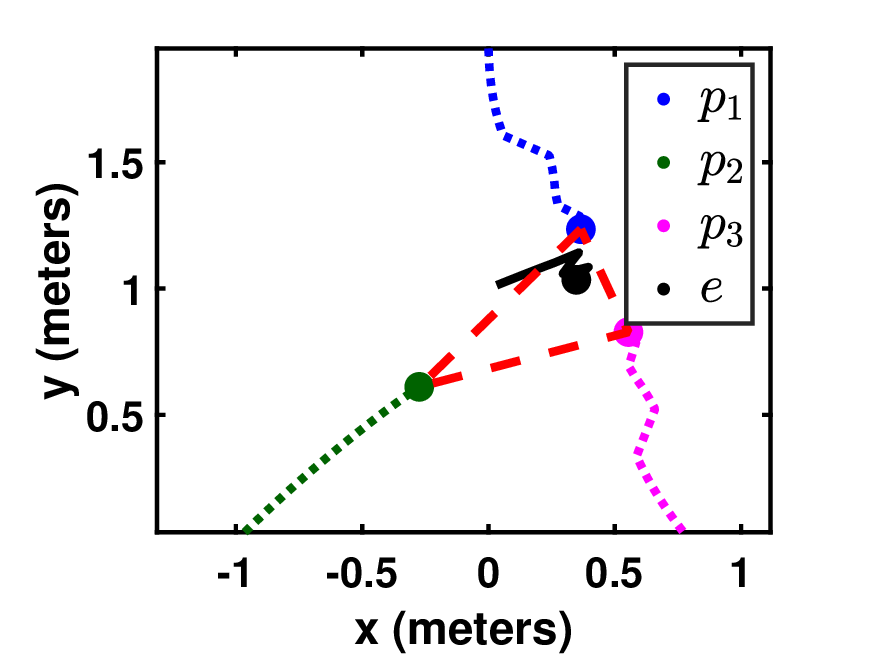}
         \caption{$t=1 \, \text{s}$}
          \label{fig3:switch_snap}
          \end{subfigure}
      \begin{subfigure}{0.24\textwidth}
       \centering 
        \includegraphics[width=\textwidth]{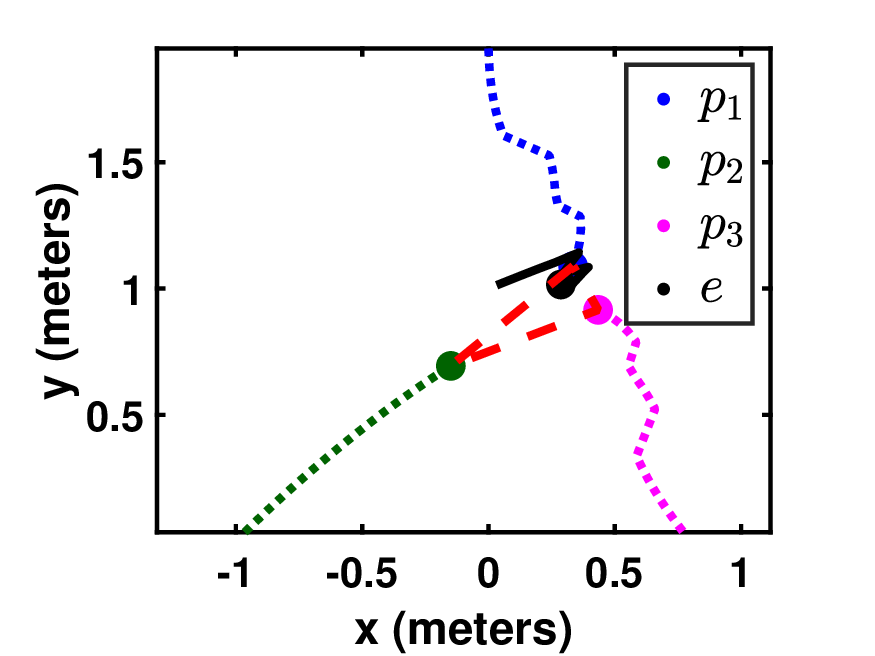}
         \caption{$t=1.2 \, \text{s}$}
          \label{fig4:switch_snap}
          \end{subfigure}
       \centering
      \caption{Encirclement guaranteed capture under the {\bf switching policy} of the evader with three pursuers. The evader moves to the closest edge by $0.5 \, \text{s}$ and switches between moving to the closest edge and moving towards the interior for every $0.15 \, \text{s}$. The evader is always encircled and is captured at $1.2 \, \text{s}$.}
      \label{fig:phase_switching}
   \end{figure*}
\begin{figure*}[t!]
      \centering
      \begin{subfigure}{0.24\textwidth}
       \centering 
         \includegraphics[width=\textwidth]{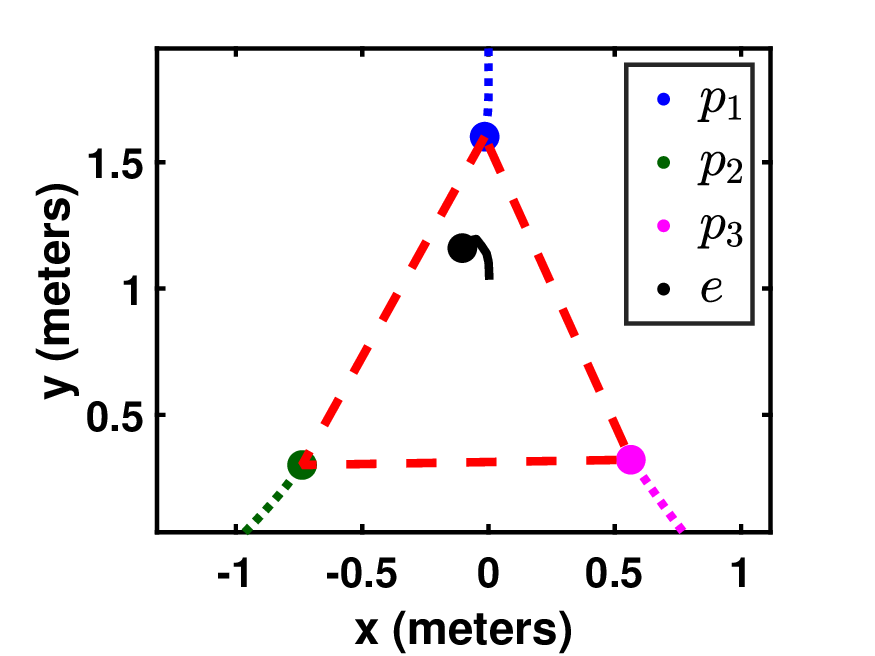}   
         \caption{$t=0.45 \, \text{s}$}
         \label{fig1:human_snap}
      \end{subfigure}
    \begin{subfigure}{0.24\textwidth}
       \centering 
     \includegraphics[width=\textwidth]{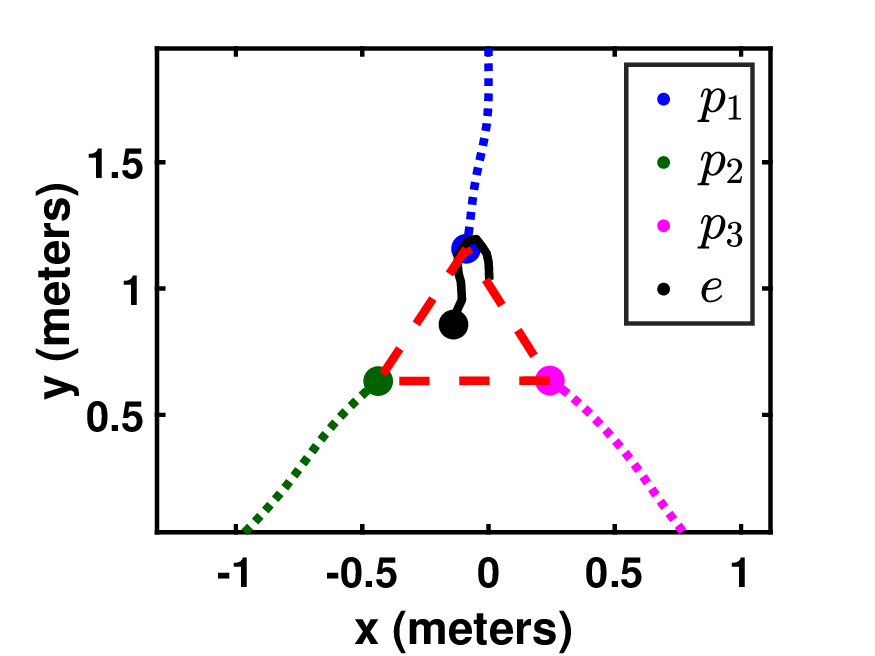}
     \caption{$t=0.9 \, \text{s}$}
      \label{fig2:human_snap}
      \end{subfigure}
    \begin{subfigure}{0.24\textwidth}
       \centering 
        \includegraphics[width=\textwidth]{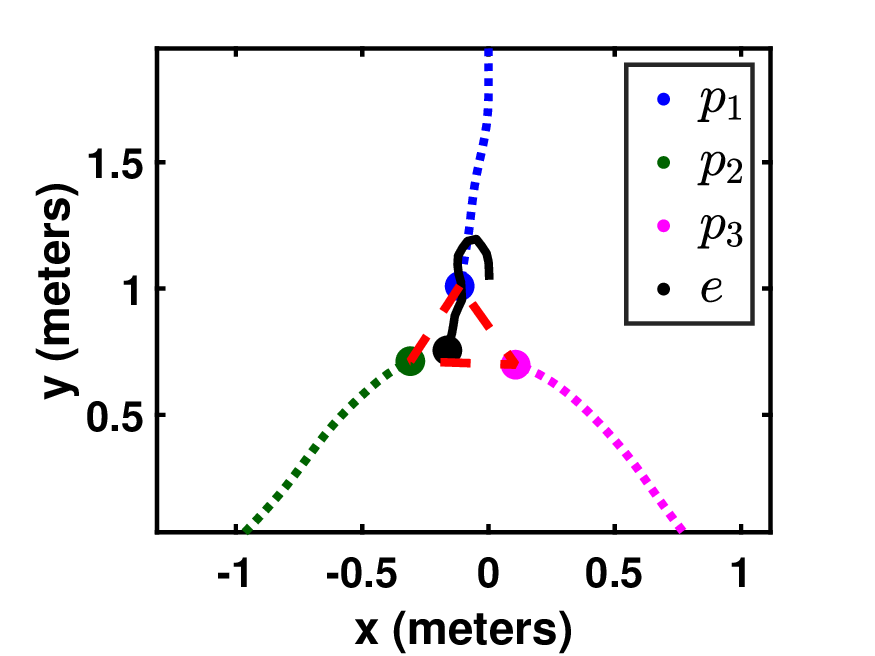}
         \caption{$t=1.05 \, \text{s}$}
          \label{fig3:human_snap}
          \end{subfigure}
      \begin{subfigure}{0.24\textwidth}
       \centering 
        \includegraphics[width=\textwidth]{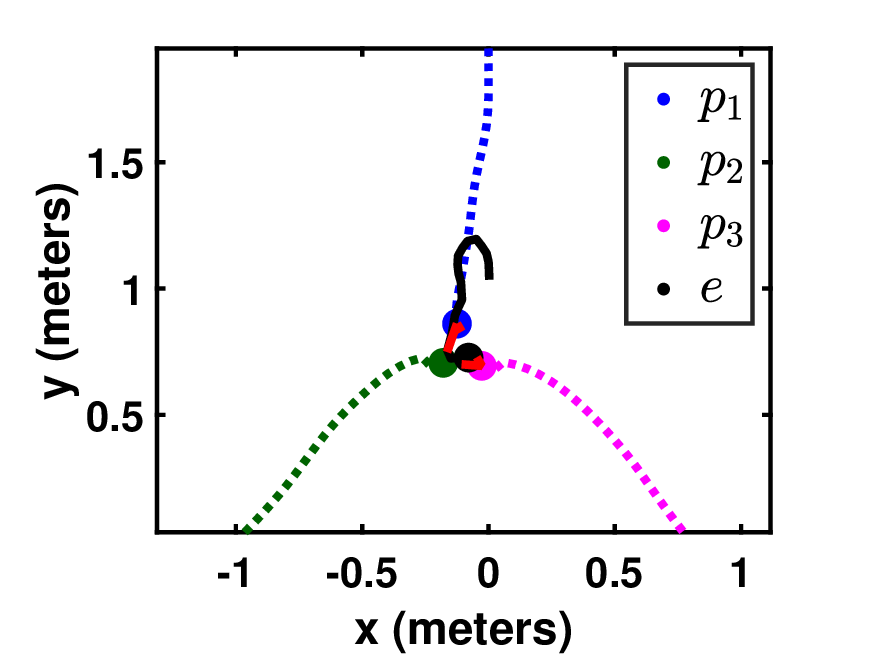}
         \caption{$t=1.19\, \text{s}$}
          \label{fig4:human_snap}
          \end{subfigure}
       \centering
      \caption{Encirclement guaranteed capture with three pursuers under {\bf random policies} adopted by a human-like evader. The evader remains encircled and is captured at $1.19 \, \text{s}$.}
      \label{fig:phase_human}
   \end{figure*}
\section{Numerical Results} \label{section:results}
We now demonstrate the effectiveness of the proposed strategies in guaranteeing encirclement and capture via numerical simulations. All simulations have been carried out in the MATLAB environment. The capture radius is taken as $r_c=0.3 \,\text{m}$. The pursuers' speed has been normalized to unity, and the evader is considered to move at its maximum speed $\mu_m<1$, which is set to be $0.7 \,\text{m/s}$. We consider three pursuers and one evader in $\mathbb{R}^2$. The pursuers' initial positions are chosen to be $(0,2)$, $(-1,0)$ and $(0.8,0)$, while the evader's initial position is chosen as $(0,1)$ which is inside the convex hull. All units are in meters. 

First, we consider the following commonly adopted worst-case evader strategies, while the pursuers follow the strategy presented in Theorem \ref{theorem:enc_capture} with $\varphi_j= \varphi_k= \sin^{-1}(\mu_m)$ for the active pursuers.
\begin{enumerate}
 \item \textbf{Greedy policy (GP)}: The evader reacts only to the nearest pursuer and moves away from it \cite{ref_slow_encirclement_wang2023distributed},\cite{ref_RMPC_patra20253d}.
 \item \textbf{Switching policy (SP)}: The evader first moves to the closest edge and then alternates between remaining on the edge and moving to the interior at a given frequency, which taken as $0.3 \, \text{s}$ ($0.15 \,\text{s}$ in each phase) until capture. 
 \item \textbf{Random policy (RP)}: We consider a human-like evader in this case, which may behave erratically and select unpredictable moving directions \cite{ref_RMPC_patra20253d,ref_rmpc_EGP_wang2021encirclement,ref_slow_encirclement_wang2023distributed}. We choose the evader's heading manually in an online manner to simulate this case. 
 \item \textbf{Stationary policy (STP)}: The evader remains stationary at its point of initialization ($\mu=0$). 
 \item \textbf{Closest Link policy (CLP)}: The evader moves towards its nearest edge at every instant and retains the edge phase until capture by choosing appropriate headings.
\end{enumerate}

\begin{figure*}[t!]
      \centering
      \begin{subfigure}{0.24\textwidth}
       \centering 
         \includegraphics[width=\textwidth]{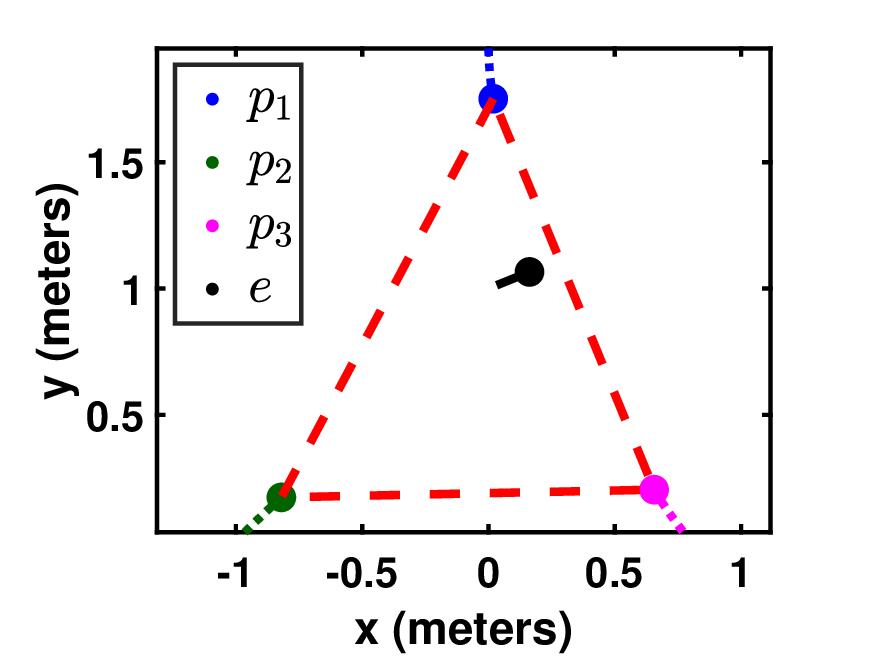}   
         \caption{$t=0.3 \, \text{s}$}
         \label{fig1:4P_snap}
      \end{subfigure}
      \begin{subfigure}{0.24\textwidth}
       \centering 
     \includegraphics[width=\textwidth]{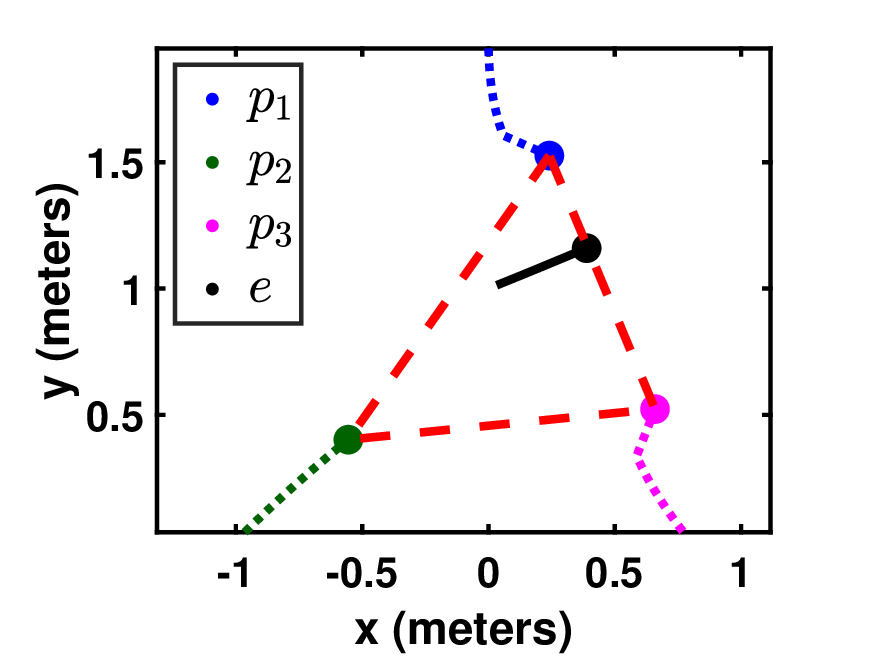}
     \caption{$t=0.65 \, \text{s}$}
      \label{fig2:4P_snap}
      \end{subfigure}
    \begin{subfigure}{0.24\textwidth}
       \centering 
     \includegraphics[width=\textwidth]{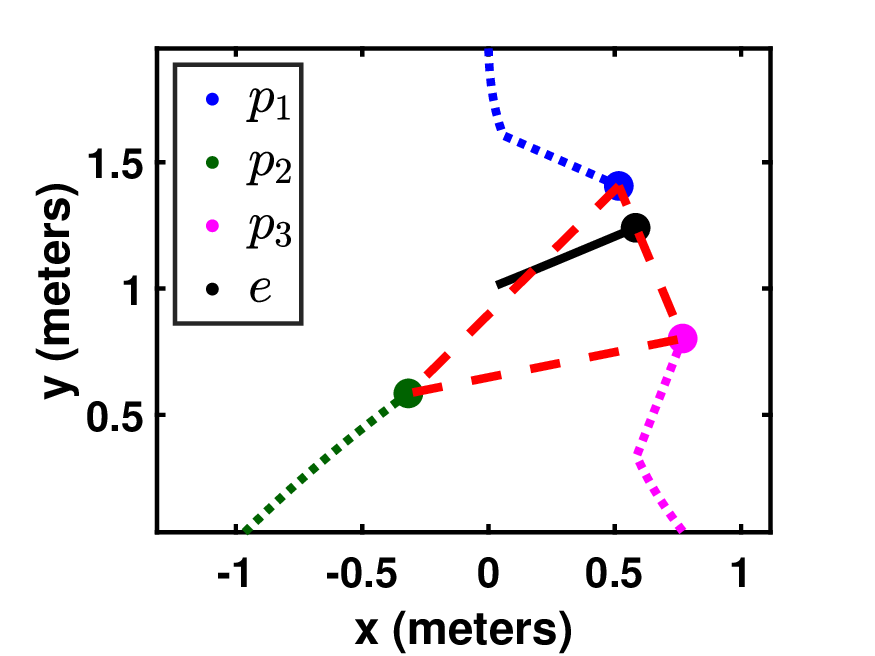}
     \caption{$t=0.95 \, \text{s}$}
      \label{fig3:4P_snap}
      \end{subfigure}
       \begin{subfigure}{0.24\textwidth}
       \centering 
        \includegraphics[width=\textwidth]{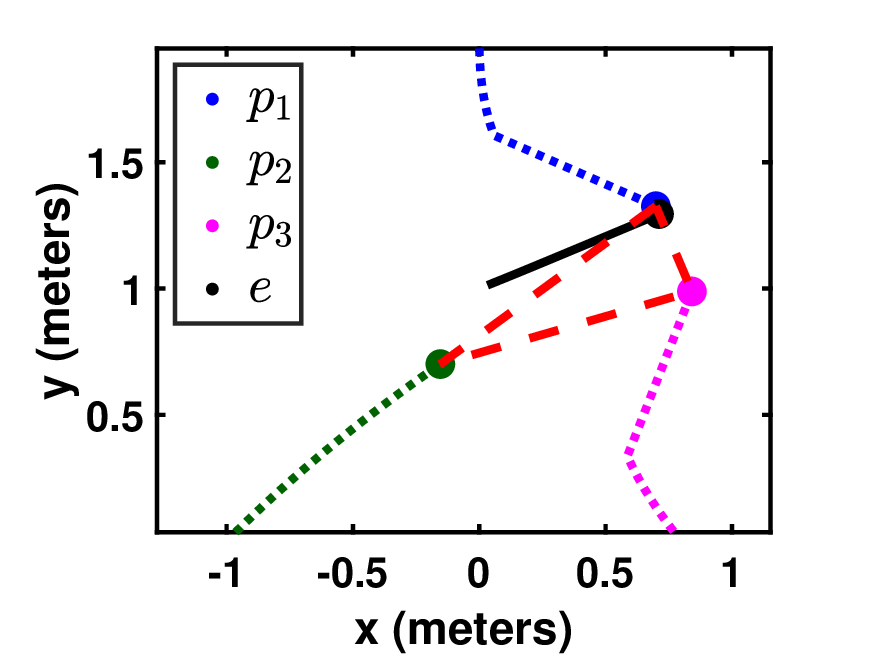}
         \caption{$t=1.15 \, \text{s}$}
          \label{fig4:4P_snap}
          \end{subfigure}
       \centering
      \caption{Encirclement guaranteed capture under the {\bf closest link policy} of the evader with three pursuers. The evader moves to the closest edge and remains on the edge until it is captured at $t=1.15 \, \text{s}$.}
      \label{fig:4PCL_phase}
   \end{figure*}
\begin{figure*}[t!]
      \centering
      \begin{subfigure}{0.24\textwidth}
       \centering 
         \includegraphics[width=\textwidth]{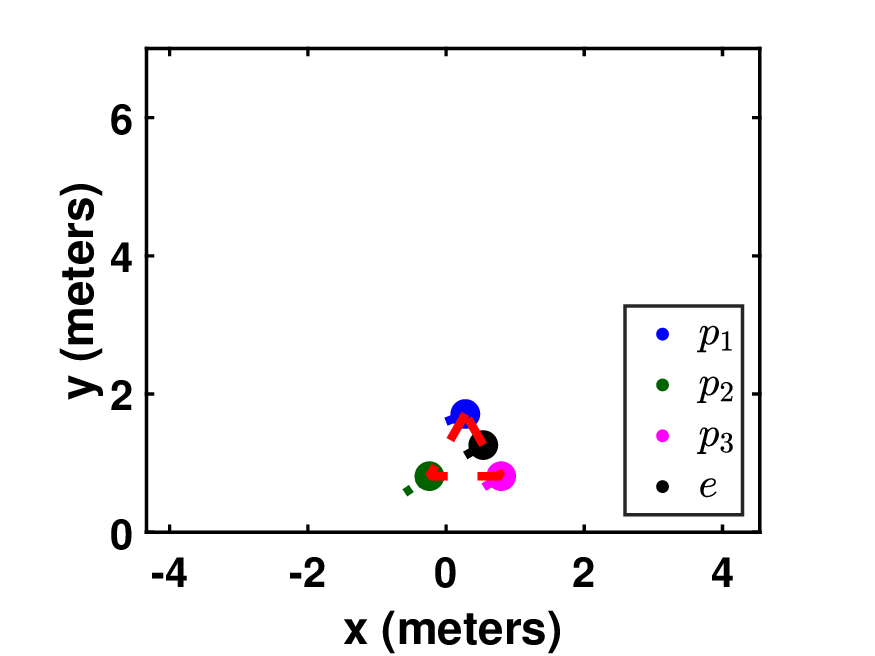}   
         \caption{$t=0.5 \, \text{s}$}
         \label{fig1:rmpc}
      \end{subfigure}
      \begin{subfigure}{0.24\textwidth}
       \centering 
     \includegraphics[width=\textwidth]{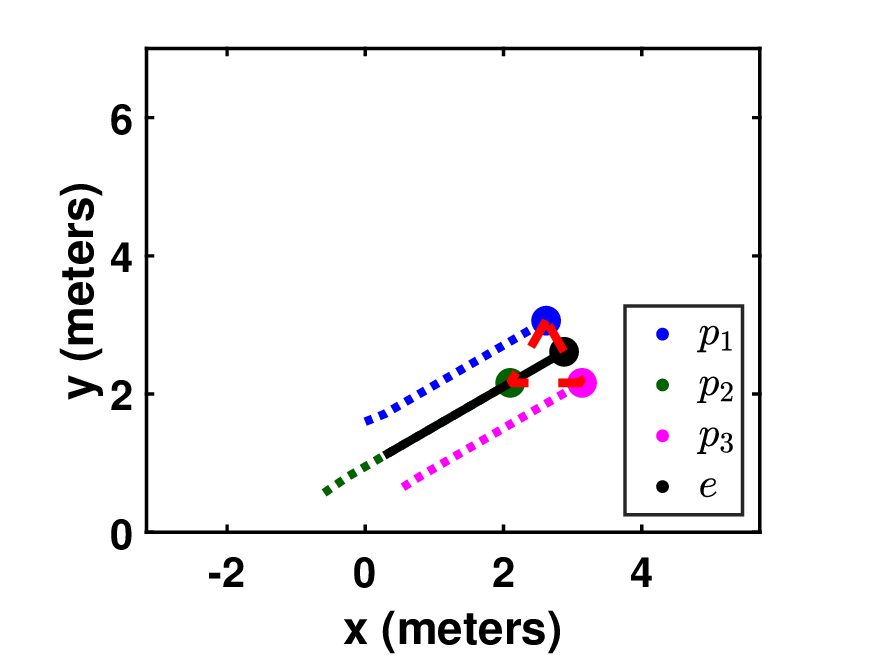}
     \caption{$t=12 \, \text{s}$}
      \label{fig2:rmpc}
      \end{subfigure}
    \begin{subfigure}{0.24\textwidth}
       \centering 
     \includegraphics[width=\textwidth]{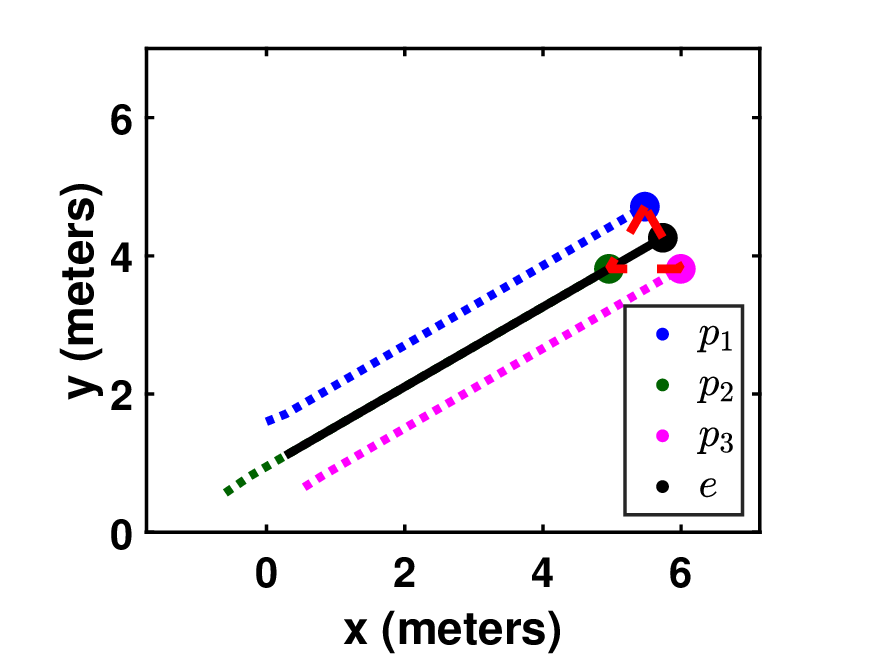}
     \caption{$t=25 \, \text{s}$}
      \label{fig3:rmpc}
      \end{subfigure}
       \begin{subfigure}{0.24\textwidth}
       \centering 
        \includegraphics[width=\textwidth]{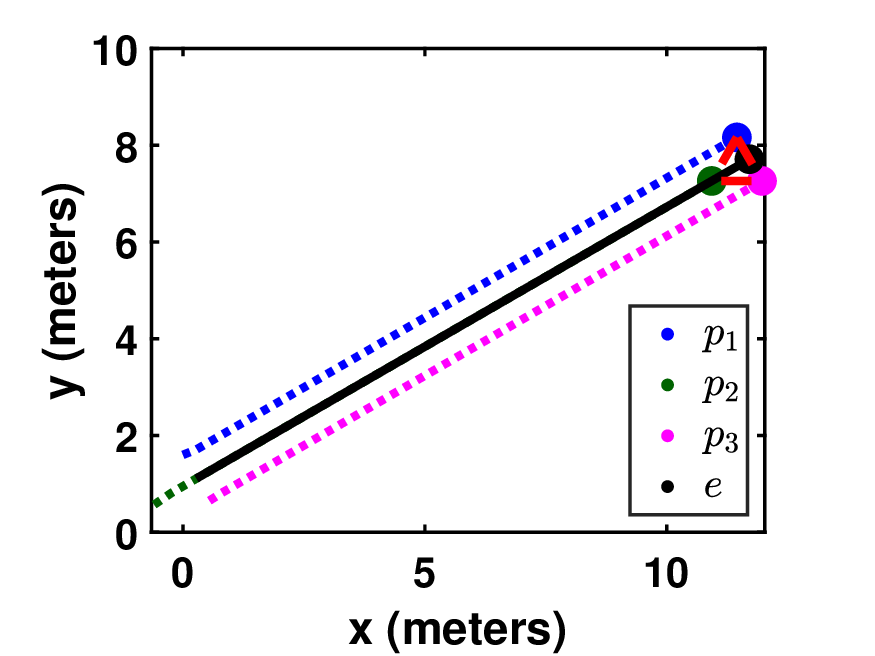}
         \caption{$t=42 \, \text{s}$}
          \label{fig4:rmpc}
          \end{subfigure}
       \centering
      \caption{Performance under the {\bf closest link policy} of the evader with the encirclement guaranteed sector set based robust MPC approach in \cite{ref_slow_RMPC_encirclement_wang2025distributed}. The evader moves to the closest edge and remains on the edge. While the pursuers are able to maintain encirclement, it takes around  $t=42 \, \text{s}$ for capture, due to the rigid sector constraints.}
      \label{fig:egss_rmpc}
   \end{figure*}
\begin{figure}
    \centering
    \includegraphics[width=0.8\linewidth]{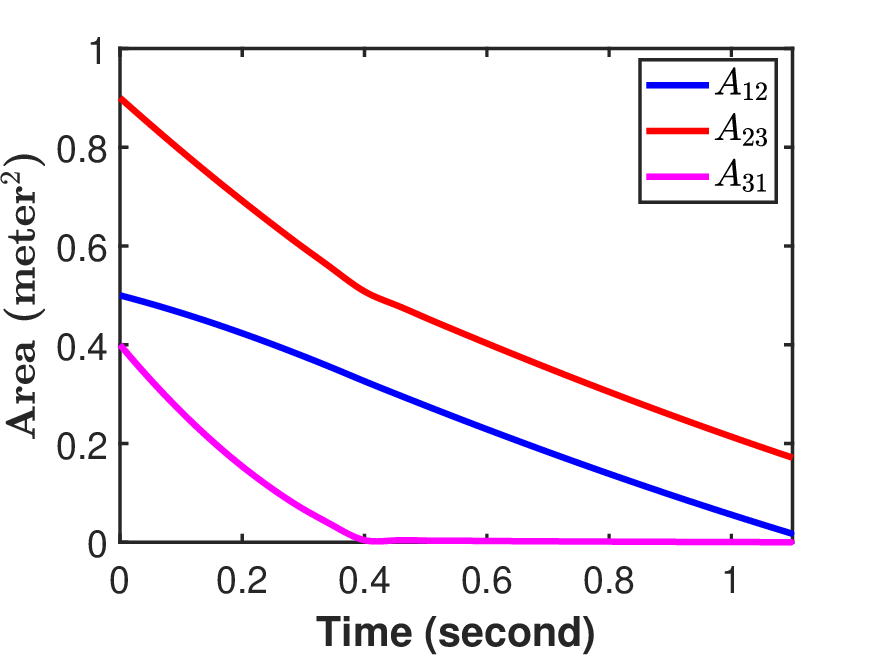}
    \caption{Evolution of the areas of the subtriangles $A_{12}$, $A_{23}$, and $A_{31}$ formed by the pursuer and evader positions during pursuit under the closest link policy of the evader. The areas remain non-negative at all times, confirming that encirclement is maintained. Furthermore, $A_{31}=0$ from approximately $0.4 \, \text{s}$ onward, indicating that the edge $p_3p_1$ remains active for the rest of the pursuit.}
    \label{fig:area_comparison}
\end{figure}
\begin{figure}
    \centering
    \includegraphics[width=0.8\linewidth]{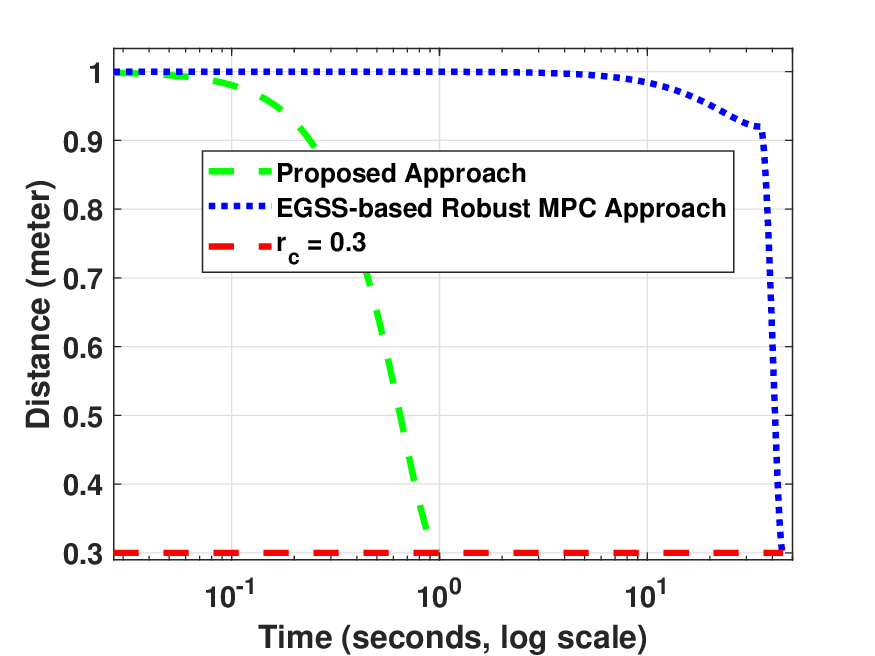}
    \caption{Distance between the nearest pursuer and the evader over time under the proposed strategy and the EGSS-based robust MPC approach in \cite{ref_slow_RMPC_encirclement_wang2025distributed}. 
    The proposed approach achieves capture much faster compared to the approach in \cite{ref_slow_RMPC_encirclement_wang2025distributed}.}
    \label{fig:dist_comparison}
\end{figure}
\begin{figure*}[t!]
      \centering
      \begin{subfigure}{0.45\textwidth}
       \centering 
         \includegraphics[width=\textwidth]{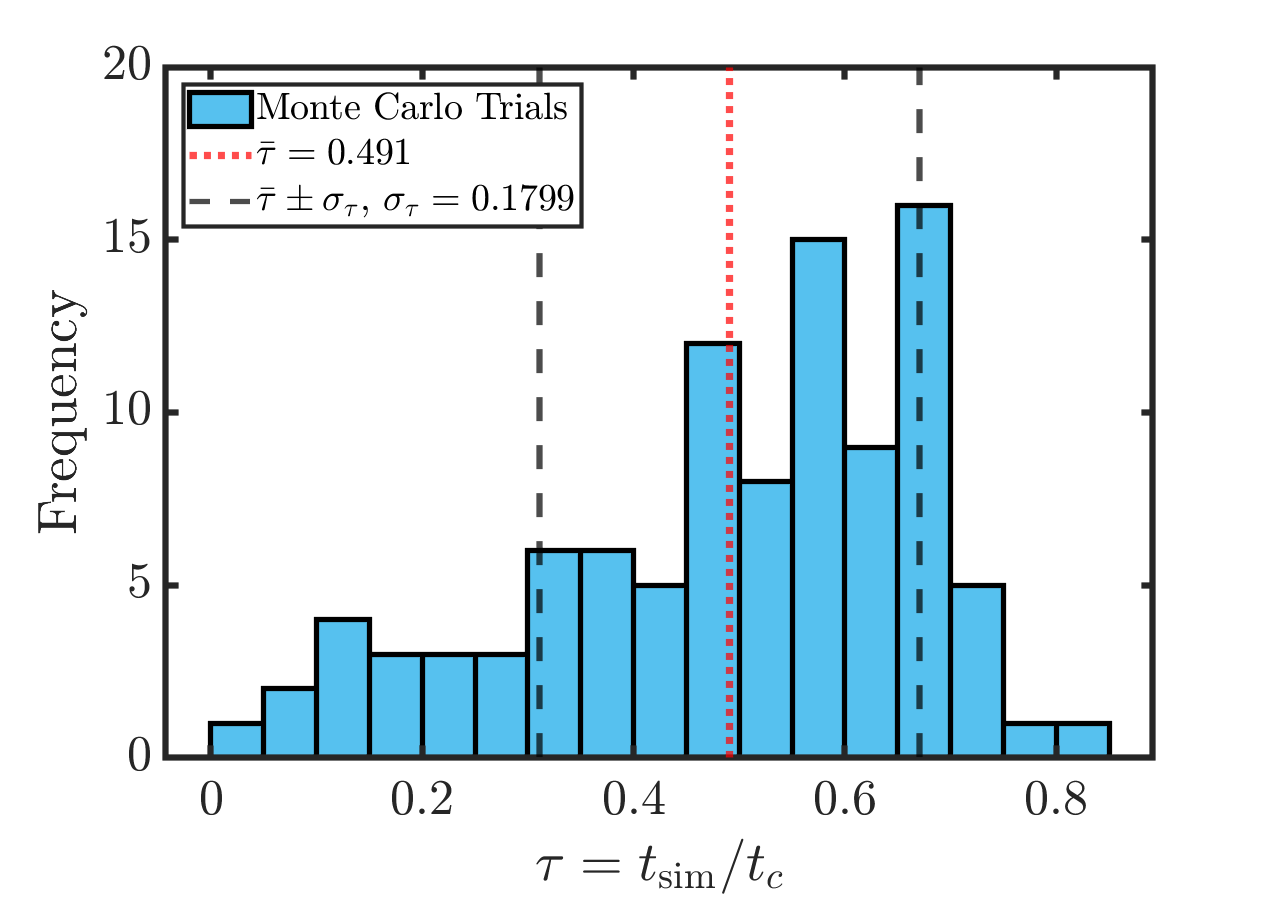}   
         \caption{$\mu_m=0.2$}
         \label{fig1:mu_0.2}
      \end{subfigure}
      ~~~~
      \begin{subfigure}{0.45\textwidth}
       \centering 
     \includegraphics[width=\textwidth]{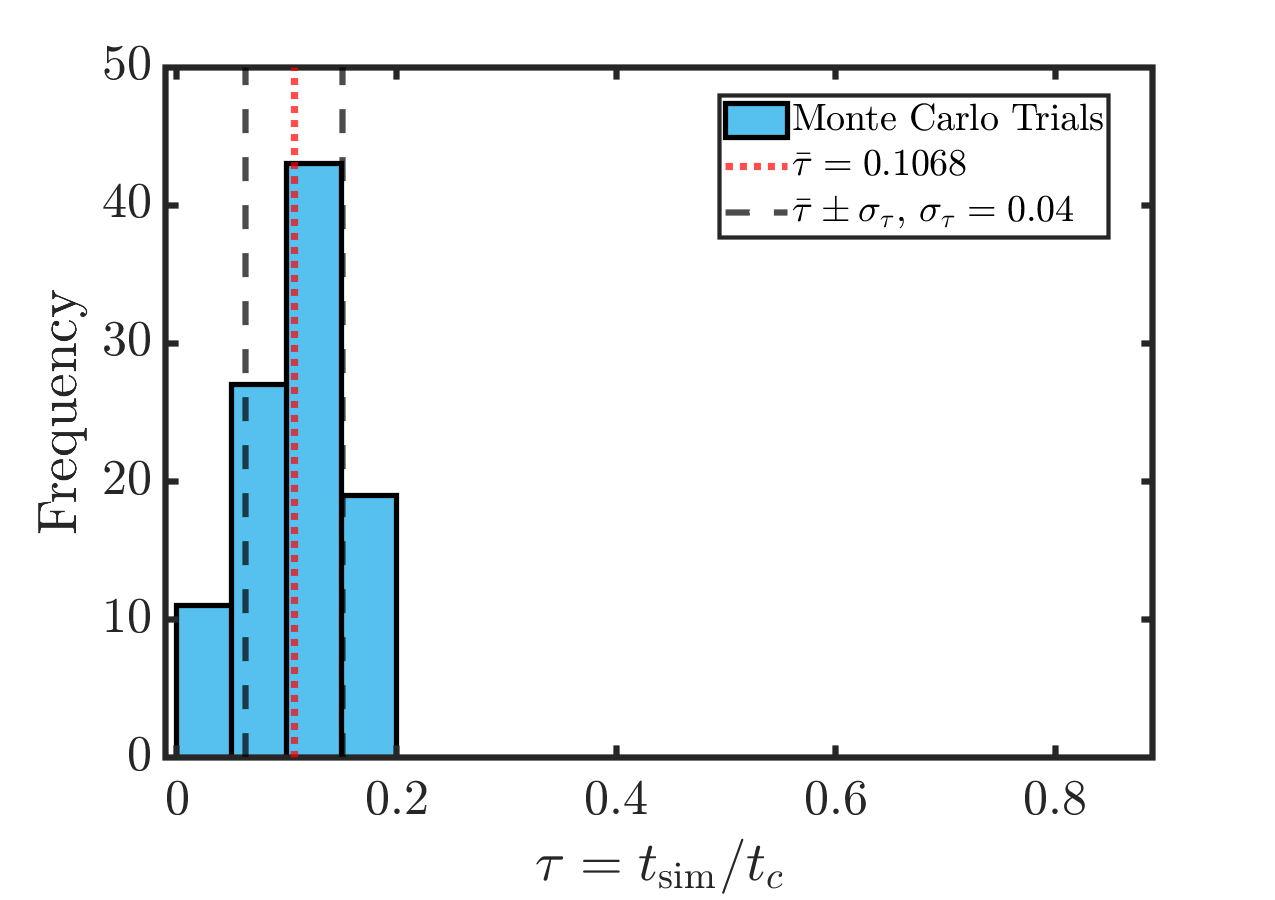}
     \caption{$\mu_m=0.9$}
      \label{fig2:mu_0.9}
      \end{subfigure}
    \centering
    \caption{Empirical distribution, mean and standard deviation of the ratio ($\tau$) of the numerical capture time ($t_{\text{sim}}$) and theoretical upper bound ($t_c$) under the closest link policy of the evader with pursuers initialized at $(0,2)$, $(-1,0)$, and $(0.8,0)$, while the evader's initial position is drawn uniformly at random from the interior of the triangle formed by the pursuers. It is found that $\tau<1$ ($t_{\text{sim}} < t_c$) in each case and the bound is more conservative for $\mu_m=0.9$ compared to $\mu_m=0.2$ in line with the theoretical findings.}
    \label{fig:capture_monte_carlo}
   \end{figure*}

Figures \ref{fig:phase_greedy} -- \ref{fig:4PCL_phase} present a few snapshots of the evolution of the pursuers' and evader's positions when the evader follows the GP, SP, RP, and CLP respectively. It is evident that the pursuers successfully maintain the evader in the convex hull and finally achieve capture in each case. The trajectories of the evader in Figure \ref{fig:phase_greedy} indicate that it is moving away from the nearest pursuer at all times. The trails of the pursuers' and the evader's trajectories in Figure \ref{fig:phase_switching} indicate that the evader is switching frequently between the nearest edge and the interior and thus making the active pursuers switch their strategies between the interior phase and the edge phase. The trajectories shown in Figure \ref{fig:phase_human} in the case of RP depend completely on the evader's input heading. It is clear from Figure \ref{fig:4PCL_phase} that under the CLP, the evader moves towards the closest edge, $p_3p_1$ and the edge gets activated by $0.4$ seconds. The evader then retains the edge phase until it is captured at $1.15$ seconds. The evolution of the areas $A_{12}$, $A_{23}$, and $A_{31}$ shown in Figure \ref{fig:area_comparison} confirms that encirclement is maintained, as the areas remain non-negative throughout the pursuit. Furthermore, $A_{31}=0$ from approximately $0.4 \, \text{seconds}$ onward, indicating that the edge $p_3p_1$ remains active for the remainder of the pursuit.

The time taken by the pursuers to capture the evader under various policies is reported in Table~\ref{tab:capture_times} against their respective upper bounds established in \eqref{eq:common_bound}. Under the stationary policy (STP), the bound is obtained by assuming $\mu_m=0$, which yields the smallest possible upper bound.\footnote{Note from \eqref{eq:common_bound} that the bound is an increasing function of $\mu_m$. Thus, if $\mu=0$ were known in advance, that would yield the smallest possible upper bound.} For all other policies the bound is obtained by setting $\mu_m=0.7$. It is seen that the interaction terminates in each case with a win for the pursuers, as the evader ends up in the capture disc of one of the pursuers. Thus, the proposed approach simultaneously guarantees encirclement and finite-time capture. 

Next, we compare our proposed approach with the encirclement guaranteed sector set (EGSS)-based robust MPC approach proposed in \cite{ref_slow_RMPC_encirclement_wang2025distributed} under the closest link evasion policy. It is seen from Figures \ref{fig:egss_rmpc} and \ref{fig:dist_comparison} that the EGSS-based robust MPC approach takes approximately $42$ seconds to achieve capture, whereas our proposed approach achieved capture in approximately $1.15$ seconds. Additionally, we conduct Monte Carlo simulations over 100 trials, where the evader’s initial position in each trial is randomly generated inside the convex hull of three pursuers to evaluate the performance statistically. For a fair comparison, the same set of randomly generated initial conditions is used for both the proposed method and the EGSS-based MPC approach. All simulations are performed under the closest link policy of the evader. The mean and standard deviation of the time to capture under the proposed method were found to be $1.04$ seconds and $0.12$ seconds, while analogous quantities were $40.34$ seconds and $2.32$ seconds under the EGSS-based robust MPC approach. This behavior can be attributed to the sector constraints, which may limit the pursuers’ ability to move directly toward the evader during the edge phase, thereby increasing the capture time.

\begin{table}[t!]
\centering
\caption{Comparison of Theoretical Bounds and Numerical Capture Times (in seconds). The bound for the stationary policy is obtained by assuming $\mu_m=0$ (which results in the tightest bound), while for all other cases, $\mu_m=0.7$.}
\label{tab:capture_times}
\begin{tabular}{|l|c|c|}
\hline
\textbf{Policy} & \textbf{Theoretical Bound} & \textbf{Numerical Capture Time} \\ \hline
Greedy & $3.10$ & $1.25$ \\
Switching & $3.10$ & $1.2$ \\
Random & $3.10$ & $1.19$ \\
Stationary & $0.93$ & $0.5$ \\
Closest Link & $3.10$ & $1.15$ \\
\hline
\end{tabular}
\end{table}

Furthermore, we numerically validate the upper bound on capture time by performing Monte Carlo simulations with $100$ trials under the closest link policy of the evader. In each trial, the pursuers are initialized at their fixed positions $(0,2)$, $(-1,0)$, and $(0.8,0)$, while the evader's initial position is drawn uniformly at random from the interior of the triangle formed by the pursuers. For each trial, we compute (i) the actual capture time $t_{\text{sim}}$ by numerically simulating the dynamics under the proposed strategies, and (ii) the theoretical upper bound given by $t_c = \frac{V_0-nr_c}{n(1-\mu_m)}$, where $V_0$ is the value of the Lyapunov function at $t=0$. The empirical distributions or histograms of the ratio $\tau = \frac{t_{\text{sim}}}{t_c}$, are shown in Figure \ref{fig:capture_monte_carlo} for $\mu_m = 0.2$ (left panel) and $\mu_m = 0.9$ (right panel). Figure \ref{fig:capture_monte_carlo} shows that $\tau < 1$ is always satisfied, i.e., the theoretical upper bound is larger than the actual capture time. Furthermore, we observe that for $\mu_m=0.9$, the bound is more conservative compared to $\mu_m=0.2$; this outcome is due to the presence of $1-\mu_m$ term in the denominator of $t_c$ which leads to $t_c$ becoming large when $\mu_m$ approaches $1$.

\section{Conclusion} \label{section:conclusion}
In this work, we presented a strategy for a group of pursuers to enclose an evader within their convex hull and capture it in finite time, without knowledge of the control strategy of the evader. Furthermore, we derived an upper bound on the time of capture that holds for all admissible evader strategies. The proposed framework assumes that the evader is initially encircled by the pursuers. Intuitively, this can potentially be achieved by exploiting the pursuers’ speed advantage to create a sufficiently large convex hull around the evader while avoiding capture; a formal analysis of this initialization phase remains an interesting direction for future research. The assumption $\mu_m<1$ is essential for the finite-time capture guarantee established in Theorem~$2$, and extending the analysis to equal-speed or faster evader scenarios with sufficiently many pursuers having capture discs constitutes an interesting open problem. Furthermore, extending the present setting to incorporate the presence of noise, time delay, and collision avoidance constraints among the pursuers remains as an important direction for future research.

\section*{Acknowledgement}
The authors thank Dr. Alexander Von Moll and Dr. Isaac Weintraub for helpful discussions.

\bibliographystyle{IEEEtran}
\bibliography{egc_dp}

\end{document}